\documentclass[12pt]{article}
\usepackage{graphicx}
\usepackage{authblk}
\usepackage{amssymb,amsmath,amsthm}
\usepackage{subfig}
\usepackage{enumerate}
\usepackage{hyperref}
\hypersetup{colorlinks,linkcolor=blue,urlcolor=blue,citecolor=blue}
\newtheorem{theorem}{Theorem}[section]
\newtheorem{proposition}[theorem]{Proposition}
\newtheorem{corollary}[theorem]{Corollary}
\newtheorem{lemma}[theorem]{Lemma}
\newtheorem{example}[theorem]{Example}
\newtheorem{remark}[theorem]{Remark}

\newtheorem{definition}[theorem]{Definition}
\newtheorem{assumption}[]{Assumption}

\makeatletter
\def\th@plain{%
  \thm@notefont{}
  \itshape 
}
\def\th@definition{%
  \thm@notefont{}
  \normalfont 
}
\makeatother
\begin{document}
 \graphicspath{{./figures/}}

\title{Relative equilibria  for Scaling Symmetries and Central Configurations}
\author[1]{Giovanni Rastelli}
\author[2]{Manuele Santoprete}
\affil[1]{Dipartimento di Matematica, Universit\`a di Torino}

\affil[2]{Department of Mathematics, Wilfrid Laurier University, Waterloo ON, N2L 3C5, Canada}
    \makeatletter
    \renewcommand\AB@affilsepx{: \protect\Affilfont}
    \makeatother

    \affil[ ]{Email addresses:}

    \makeatletter
    \renewcommand\AB@affilsepx{, \protect\Affilfont}
    \makeatother
    \affil[1]{giovanni.rastelli@unito.it}
    \affil[2]{msantoprete@wlu.ca}
\setcounter{Maxaffil}{0}
\renewcommand\Affilfont{\itshape\small}
\maketitle

\tableofcontents
\begin{abstract}
In this paper, we explore scaling symmetries within the framework of symplectic geometry. We focus on the action \(\Phi\) of the multiplicative group \(\mathbb{R}^+\) on exact symplectic manifolds \((M, \omega,\theta)\), with \(\omega = -d\theta\), where $ \theta $ is  a given  primitive one-form. Extending established results in symplectic geometry and Hamiltonian dynamics, we introduce conformally symplectic maps, conformally Hamiltonian systems, conformally symplectic group actions, and the notion of conformal invariance. This framework allows us to generalize the momentum map to the conformal momentum map, which is crucial for understanding scaling symmetries. Additionally, we provide a generalized Hamiltonian Noether's theorem for these symmetries.

We introduce the (conformal) augmented Hamiltonian \(H_{\xi}\) and prove that the relative equilibria of scaling symmetries are solutions to equations involving $ H _{ \xi } $ and  the primitive   one-form \(\theta\). We derive their main properties, emphasizing the differences from relative equilibria in traditional symplectic actions.

For cotangent bundles, we define a scaled cotangent lifted action and derive explicit formulas for the conformal momentum map. We also provide a general definition of central configurations for Hamiltonian systems on cotangent bundles that admit scaling symmetries. Applying these results to simple mechanical systems, we introduce the augmented potential \(U_{\xi}\) and show that the relative equilibria of scaling symmetries are solutions to an equation involving  $ U _{ \xi } $ and the Lagrangian one-form \(\theta_L\).

Finally, we apply our general theory to the Newtonian \(n\)-body problem, recovering the classical equations for central configurations.
\end{abstract}

\section{Introduction}

\subsection{Background and Motivation}

Relative equilibria  have been the subject of extensive research  and  are central to  the study of mechanical systems with symmetry, often corresponding to significant physical phenomena such as steady rotations or uniform motions. 
 Relative equilibria have played a significant role in foundational works of modern geometric mechanics
such as those by Arnold \cite{arnold1966geometrie}, Smale \cite{smale1970topologyI,smale1970topologyII}, and Marsden and Weinstein \cite{marsden1974reduction}. 

A relative equilibrium can be defined in two closely related ways: either as a point in phase space where the trajectory of a mechanical system coincides with the orbit of a one-parameter symmetry group, or as a specific solution that traces out the orbit of such a group. The latter interpretation is commonly employed in the context of the Newtonian \(n\)-body problem, where relative equilibria represent configurations of point masses rotating uniformly around their center of mass while maintaining a fixed shape. For the planar case \(d = 2\), the relevant symmetry group is \(G = \text{SO}(2)\), corresponding to rotations in the plane.

In this paper, we adopt the former interpretation, treating relative equilibria as points with the symmetry properties described above. In the context of geometric mechanics, it is common to consider a finite-dimensional symplectic manifold \((M, \omega)\), on which a Lie group \(G\) (with Lie algebra \(\mathfrak{g}\)) acts symplectically, along with a \(G\)-invariant Hamiltonian \(H\). This symplectic action naturally gives rise to the equivariant momentum map \(J: M \rightarrow \mathfrak{g}^*\), which encodes the conserved quantities associated with the system's symmetries.

In this setting, there is an elegant framework for understanding relative equilibria. Relative equilibria can be obtained as the critical points of the augmented Hamiltonian 
\[
H_{\xi} = H - \left\langle \xi, J \right\rangle,
\]
where \(\xi \in \mathfrak{g}\), and $\left\langle \cdot , \cdot \right\rangle$ denotes the paring between a Lie algebra and its dual. If \(Q\) is the configuration space, for simple mechanical systems on the cotangent bundle \(M = T^*Q\), one can define the augmented potential 
\[
U_{\xi} = U(q) - \frac{1}{2} \left\langle \xi, \mathbb{I}(q) \xi \right\rangle,
\]
where \(\mathbb{I}\) is the locked inertia tensor. In this case, a point \((q_e, p_e)\) is a relative equilibrium if \(q_e\) is a critical point of the augmented potential and \(p_e\) satisfies a certain equation. A similar characterization can also be given in terms of the so-called amended potential.

Furthermore, by utilizing the conserved quantities encoded in the momentum map through the process of symplectic reduction \cite{marsden1974reduction}, one can reduce the dimensionality of the phase space, leading to a new, lower-dimensional symplectic manifold known as the reduced space. 
The dynamics on this reduced space are governed by the reduced Hamiltonian equations, and relative equilibria can then be viewed as equilibrium points of these reduced equations.

Scaling symmetries are ubiquitous in physical problems, appearing in fields such as fluid dynamics, field theory, and classical mechanics.
Despite the extensive body of work on this topic, the treatment of scaling symmetries and their corresponding reduction in classical mechanics has historically been approached on a case-by-case basis, focusing on specific systems (see \cite{bravetti2023scaling} for some examples).

Bravetti, Jackman, and Sloan \cite{bravetti2023scaling} have recently shown that a symplectic Hamiltonian system with scaling symmetry can be reduced to a contact Hamiltonian system with one fewer degree of freedom. This reduction depends on a scaling function, which plays a role similar to equivariant momentum maps in symplectic reduction. A key feature of their approach is its flexibility, allowing for various choices of scaling functions to obtain the reduced equations. In systems on cotangent bundles with a Riemannian metric on the configuration space, there is often a natural choice for a scaling function, something that is implicitly apparent in \cite{bravetti2023scaling}. For instance, in the two-body problem, a power of the configuration variable is a natural choice, while in the $n$-body system, a power of the moment of inertia serves this purpose. In this paper, we introduce a new distinguished scaling function, which we refer to as the conformal momentum map, extending the notion of the momentum map in symplectic reduction.

Building on the contact reduction approach for scaling symmetries, it is natural to define relative equilibria as critical points of the reduced contact Hamiltonian system. However, from a practical standpoint, this approach may not be the most convenient. Therefore, it is desirable to characterize relative equilibria in a manner similar to the case of symplectic group actions on symplectic manifolds, by using the augmented Hamiltonian and augmented potential. This method provides a more direct approach to identifying relative equilibria. This motivation is a key focus of our paper.

Another significant motivation is to provide a general definition of central configurations. In the context of the \(n\)-body problem, central configurations are specific arrangements of point masses such that the acceleration vector on each mass, due to the gravitational influence of all other masses, points toward the center of mass and is proportional to the distance from the center of mass. While the study of relative equilibria and central configurations is a key focus of Celestial Mechanics (see, for instance, \cite{santoprete2021uniqueness} for an explanation of the importance of central configurations), the concept of central configurations, unlike that of relative equilibria, has not been widely generalized beyond the Newtonian \(n\)-body problem in Euclidean spaces. An exception is the work by Diacu, Stoica and Zhu \cite{diacu2018central}, which provides a definition for central configurations in the \(n\)-body problem on spaces of non-zero constant curvature. Therefore, another one of our goals is to utilize scaling symmetries to define central configurations in a more general setting.

Generalizing the notion of central configurations to broader  settings becomes both timely and relevant since there  has been a recent surge of interest in studying generalized versions of the Newtonian $n$-body problem. This includes the Kepler problem and $n$-body problems in spaces of constant curvature (e.g., \cite{diacu2012part1,diacu2012part2,diacu2012relative,diacu2014relative,diacu2018central}) and on surfaces of revolution \cite{santoprete2008gravitational,stoica2018n}. Research has also explored the Kepler problem and the two-body problem on the Heisenberg group \cite{dods2019numerical,dods2021self,montgomery2015keplerian,stachowiak2021non}, and the Anisotropic Kepler problem \cite{devaney1978collision,gutzwiller1973anisotropic}.

The approach for defining central configurations introduced in this paper  applies to various contexts, including the Kepler and $n$-body problems on surfaces of revolution that admit scaling symmetries, the Anisotropic Kepler problem, and the Kepler and two-body problem on the Heisenberg group. However, in the case of the Heisenberg group, the general theory is not fully utilized, as the scaling symmetry is symplectic, even though the Hamiltonian is conformally invariant \cite{montgomery2015keplerian}. 
Furthermore, note that the results from Section \ref{sec:rel-eq-cot-bundle} cannot be used for the Kepler and two-body problem on the Heisenberg group, as the configuration space in these cases is endowed with a sub-Riemannian metric rather than a Riemannian one.

Unfortunately, however, the approach adopted in this study associates the definition of central configurations with scaling symmetries, making it infeasible to extend this definition to spaces of constant curvature. In fact,  although the sphere and the hyperbolic plane are homogeneous and isotropic, they do not permit dilations, thereby preventing the application of our central configuration definition to these geometries.

\subsection{Summary of Main Results and Outline of the Paper}
In this paper, we  study  systems with  scaling symmetries within the framework of symplectic geometry. To effectively capture these symmetries, we consider the action \( \Phi \) of the multiplicative group \( \mathbb{R}^{+} \) on an exact symplectic manifold \( (M, \omega) \), where \( \omega = -d\theta \).

Building on established results in symplectic geometry and Hamiltonian dynamics, such as those found in \cite{abraham2008foundations}, \cite{libermann2012symplectic} and \cite{marsden1992lectures}, we extend these concepts to the realm of conformal symplectic systems.
 Our approach is partially inspired by  the study of conformal Hamiltonian systems, as discussed in \cite{mclachlan2001conformal} and \cite{mclachlan2018hamiltonian}.

Traditional concepts of symplectic group actions and function invariance are insufficient for describing scaling symmetries. Therefore, we introduce conformally symplectic maps, conformally symplectic group actions, and the notion of conformal invariance. We define \( \Phi \) as a scaling symmetry for a Hamiltonian system \( (M, \omega, H) \) if \( \Phi \) is a conformally symplectic group action and \( H \) is conformally invariant.

We also introduce conformally Hamiltonian systems on exact symplectic manifolds. These systems extend the traditional Hamiltonian framework and have been used to study certain dissipative systems characterized by specific types of friction, as investigated in \cite{mclachlan2001conformal}. However, our focus is on using conformally Hamiltonian systems to understand scaling symmetries, rather than to study  dissipative systems.

With these concepts in place, we generalize the notion of the momentum map to the conformal momentum map, which is adapted to scaling symmetries. For an element \(\xi\) of the Lie algebra of \(G\), we define the momentum map \(J\) so that the conformal momentum function \(J_{\xi} = \left\langle J, \xi \right\rangle\) is a conformal Hamiltonian, satisfying the equation:
\[
i_{\xi_M} \omega + (c\xi) \theta = dJ_{\xi},
\]
where \(\xi_M\) is the infinitesimal generator of the action \(\Phi\) associated with \(\xi\), $ \omega $ the symplectic form with 
$ \omega = - d \theta $. 

In the context of scaling symmetries, we define relative equilibria and derive their main properties. This requires careful consideration due to key differences from relative equilibria in symplectic actions. A significant property of relative equilibria for scaling symmetries is that they can be determined using the augmented Hamiltonian \( H_{\xi} = H - J_{\xi} \) through the equation:
\[
dH_{\xi} + (c\xi) \theta = 0,
\]
where \( c \) is a constant. This differs from  the case of symplectic actions, where relative equilibria are identified as the critical points of the augmented Hamiltonian.

We specialize our formulas for the momentum map to the case of the cotangent bundle. Given an action \( \Psi \) of \(\mathbb{R}^{+} \) on \( Q \), we define a scaled cotangent lifted action on \( T^*Q \). In this context, the conformal momentum map is given by:
\[
J_{\xi}(p_q) = \left\langle J(p_q), \xi \right\rangle = \left\langle p_q, \xi_Q(q) \right\rangle,
\]
where \( q \in Q \), \( p_q \) denotes a covector at \( q \), and \( \xi_Q \) is the infinitesimal generator of \( \Psi \) on \( Q \). This conformal momentum map has the special property of being a scaling function, which is crucial for the reduction studied in \cite{bravetti2023scaling}.

We then apply our results on the relative equilibria of scaling symmetries to simple mechanical systems on cotangent bundles. A simple mechanical system has a Hamiltonian that can be written as the sum of potential and kinetic energies. In this case, we introduce the {\bf amended potential} \( U_{\xi} \), defined by 
\[
U_{\xi}(q) = U(q) - \frac{1}{2} \left\langle \xi, \mathbb{I} \xi \right\rangle,
\]
where \( \mathbb{I} \) is the {\bf locked inertia tensor} (see Section \ref{sec:rel-eq-cot-bundle}). We show that the relative equilibria of the scaling symmetry are solutions to the following equations:
\[
p_q = \mathbb{F}L(\xi_Q(q)), \quad dU_{\xi} + (c \xi) \theta_L(\xi_Q(q)) = 0,
\]
where \( \mathbb{F}L \) is the {\bf Legendre transform} or {\bf fiber derivative}  (see Section \ref{sec:rel-eq-cot-bundle}), \( \theta_L \) is the {\bf Lagrangian one-form} (see Section \ref{sec:rel-eq-cot-bundle}), and \( \xi_Q \) is the infinitesimal generator of \( \Psi \) on \( Q \).

For a Hamiltonian system on the cotangent bundle \( T^*Q \) with a scaling symmetry, we define {\bf central configurations} as follows: If \( (q_e, p_e) \) is a relative equilibrium of the scaling symmetry, then \( q_e \) is a central configuration. We apply our general theory to the Newtonian \( n \)-body problem, which involves \( n \) point particles moving in \( \mathbb{R}^3 \) and interacting gravitationally, to recover the classical equations for central configurations.

This paper is organized as follows: In Section 2, we introduce conformally symplectic maps, conformally Hamiltonian systems, and their basic properties. Section 3 provides explicit coordinate equations for conformally Hamiltonian vector fields on cotangent bundles, introduces the scaled cotangent lift, and showes that cotangent lifted diffeomorphisms are conformally symplectic. These elements are essential for understanding the geometric framework of our study. 

In Section 4, we delve into conformally symplectic Lie group actions, conformally invariant functions, and scaling symmetries. For cotangent bundles, we define scaled cotangent lifted actions and present a useful coordinate formula for the infinitesimal generators of such actions. Section 5 introduces conformal momentum maps and provides a concise expression for these maps in the case of scaled cotangent lifts. Additionally, we reveal several properties of momentum maps, including a generalization of Noether's theorem and the fact that conformal momentum maps for cotangent bundles are not  equivariant but are conformally invariant.

Section 6 focuses on relative equilibria for scaling symmetries and provides a general and novel definition of central configurations. We present convenient formulas for finding this type of relative equilibria based on the augmented Hamiltonian, offering new insights and extending classical results to this broader setting. An application to the Newtonian \( n \)-body problem is also presented.

Finally, Section 7 specializes the results of Section 6 to simple mechanical systems on cotangent bundles. We introduce the augmented potential, providing a more convenient method for finding relative equilibria of the scaling symmetry and central configurations in this special case. 

\section{Conformally Hamiltonian systems}

In this section, we introduce conformally Hamiltonian systems, focusing our discussion on exact symplectic manifolds. A symplectic manifold \((M, \omega)\) consists of a smooth manifold \(M\) and a closed, non-degenerate 2-form \(\omega\). An {\bf exact symplectic manifold} is a particular case where the symplectic form \(\omega\) is exact, meaning there exists a 1-form \(\theta\), known as the primitive one-form or symplectic potential, such that \(\omega = -d\theta\).

Note that if \(\omega\) is exact, any 1-form \(\theta' = \theta + \lambda\) with \(d\lambda = 0\) is also a primitive one-form, indicating that there are multiple 1-forms \(\theta'\) satisfying \(\omega = -d\theta'\). For an exact symplectic manifold where a specific symplectic potential has been chosen, we will use the notation \((M, \omega, \theta)\). In this paper, for simplicity,  we will always assume that a choice of symplectic potential has been made.


\begin{definition}[{\bf Conformally Symplectic map}] Let $ (M_1 , \omega_1,\theta_1) $  and $ (M_2, \omega_2,\theta_2) $  be  exact symplectic manifolds. A $ C ^{ \infty } $ mapping $ f: M_1\to M_2  $  is called {\bf conformally symplectic}  with parameter $ c $  if $f ^\ast \omega_2 = c \omega_1$ for some $ c \in \mathbb{R}  ^{+ } $. 
\end{definition}
In the literature a range of terms are utilized to describe these mappings, \cite{meyer2017introduction} uses ``symplectic transformations with multiplier $ c $" or ``symplectic scalings", \cite{mclachlan2018hamiltonian} employs ``conformal symplectic maps", while \cite{Carinena_Canonoid_2013} refers to them as  ``non-strictly canonical".  The term ``conformally symplectic" has also been  previously employed, as seen in \cite{calleja2013kam}.
\begin{definition}[{\bf Conformally Hamiltonian Systems}] Let $ (M , \omega,\theta) $ be an exact symplectic manifold   and $ F: M \to \mathbb{R}  $ a $ C ^{ \infty } (M) $ function. We say that  a vector field $ X _F ^c $ is a {\bf conformally Hamiltonian vector field} with parameter $ c $   provided that 
	\[ i _{ X _F ^c } \omega + c \theta = dF\]
	We call the function $ F $ the {\bf conformal Hamiltonian}, and $ ( M , \omega ,\theta, X _F ^c) $ a {\bf conformal Hamiltonian system} with parameter $ c $. 

\end{definition}

The following proposition generalizes the conservation of energy for Hamiltonian systems.

\begin{proposition}\label{prop:dF/dt} Let $ (M, \omega ,\theta, X _F ^c) $ a conformal Hamiltonian system with parameter $ c$, and $ \gamma (t) $ an integral curve for $ X _F ^c $. Then \[\frac{ d F (\gamma (t))    } { dt } = c  i _{ X _F^c } \theta. \]	 
\end{proposition}
\begin{proof} 
	\begin{align*} \frac{ d } { dt } F (\gamma (t)) & = dF (\gamma (t)) \cdot \gamma' (t) \\
	                                          & = dF (\gamma (t)) \cdot X _{ F } ^c (\gamma (t)) \\
						  & =  i _{ X _F ^c } dF (\gamma (t)) \\
						  & = c i _{  X _F ^c } \theta + i _{ X _F ^c } i _{ X _F ^c } \omega  \\
						  & = c i _{ X _F ^c } \theta 
	\end{align*} 
\end{proof} 
Recall that the if $ (M, \omega,\theta) $ is a $ 2n$-dimensional symplectic manifold, then  the Liouville volume is 
\[\Lambda = \frac{ (- 1) ^{ n /2 }} { n! } \omega \wedge \cdots \wedge \omega \quad \mbox{(n times)}.  \]
\begin{proposition}
Let $ (M , \omega ,\theta, X _F ^c) $ a conformally Hamiltonian system  with parameter $ c $, and $ \phi _t $ be the flow of $ X _F ^c $. Then 
 for each $t$,  $ \phi _t ^\ast \omega = e ^{ ct } \omega $, that is, $ \phi _t $ is conformally symplectic for each fixed value of $t$. Moreover  the phase space volume $ \Lambda $ satisfies the equation $ \phi _t ^\ast\Lambda  = e ^{ nct } \Lambda $.
\end{proposition} 
\begin{proof} 
Note that 
\[ \frac{ d } { dt } \left( \phi _t ^\ast \omega  \right) = \left.\frac{ d } { ds } \right| _{ s = t }  \phi _s ^\ast \omega=
	\left.\frac{ d } { ds } \right| _{ s = 0 }  \phi _t ^\ast \phi _s ^\ast  \omega=\phi _t ^\ast \left.\frac{ d } { ds } \right| _{ s = 0 }  \phi _s ^\ast \omega= \phi _t ^\ast \mathcal{L} _{ X _F ^c  } \omega \]
Hence, 
\begin{align*} 
	\frac{ d}{dt} (\phi ^\ast _t \omega ) & = \phi _t ^\ast \mathcal{L} _{ X _F ^c  } \omega\\
				    & = \phi _t ^\ast ( i _{ X _F ^c } d \omega + d (i _{ X _F ^c } \omega)  )\\
				    & = \phi _t ^\ast (0 + d (dF - c \theta ) )\\
				    & = \phi _t ^\ast (c\omega), \\
				    & = c\phi _t ^\ast (\omega)  
\end{align*} 
with the initial condition $ \phi _{ t = 0 } ^\ast \omega = \omega $. Hence, the unique solution of the intial value problem is $ \phi _t ^\ast \omega = e ^{ ct } \omega $. Using this we can compute the following

\begin{align*}  \frac{ d } { dt } \left( \phi _t ^\ast \Lambda\right) & = \phi _t ^\ast \mathcal{L} _{ X _F ^c  } \Lambda \\
& = \phi _t ^\ast \left( \frac{ (- 1) ^{ n /2 }} { n! } \mathcal{L} _{ X _F ^c } (\omega \wedge \cdots \wedge \omega )\right) \\
& = \phi _t ^\ast \left( \frac{ (- 1) ^{ n /2 }} { n! } \left((\mathcal{L} _{ X _F ^c } \omega) \wedge \cdots \wedge \omega+ \omega \wedge (\mathcal{L} _{ X_F ^c } \omega ) \wedge \cdots \wedge \omega + \omega \wedge \omega \wedge \cdots \wedge( \mathcal{L} _{ X _F ^c } \omega)  \right) \right) \\
& = 	\phi _t ^\ast \left( \frac{ (- 1) ^{ n  /2 }} { n! } (c+c+ 
	\cdots + c)  \omega\wedge \cdots \wedge \omega \right)\\
& = \phi _t ^\ast \left( n c \Lambda \right), 
\end{align*} 
which gives $ \phi _t ^\ast \Lambda  = e ^{ nct } \Lambda $. 
\end{proof} 
\begin{definition} 
	A vector field $ X ^c $ on an exact  symplectic manifold $ (M, \omega,\theta) $ is called {\bf locally conformally Hamiltonian} with parameter $ c \in \mathbb{R}  $  if the one-form   $ i _{ X ^c } \omega + c \theta $ is closed, that is $ d(i _{ X ^c } \omega + c \theta) = 0  $. 
\end{definition} 
\begin{remark}
Note that the concept of a locally conformally Hamiltonian vector field seems to be  closely related to the notion of dynamical similarity, introduced by David Sloan in \cite{sloan2018}.
\end{remark} 
\begin{proposition}A vector field $ X ^c $ on an exact  symplectic manifold $ (M, \omega,\theta) $ is  locally conformally Hamiltonian with parameter $ c \in \mathbb{R}   $ 	if and only if $ \mathcal{L} _{ X ^c } \omega = c \omega $.  
\end{proposition} 
\begin{proof}
Since \( d \omega = 0 \), we have \( \mathcal{L}_{X^c} \omega = d (i_{X^c} \omega) + i_{X^c} (d \omega) = d (i_{X^c} \omega) \). 
Suppose,  $ d(i _{ X ^c } \omega + c \theta) = 0  $, then 
\[
d(i_{X^c} \omega + c \theta) = d (i_{X^c} \omega) + c d \theta = \mathcal{L}_{X^c} \omega - c \omega = 0.
\]
Conversely, if $ d(i_{X^c} \omega + c \theta) =0 $, then we have 
\[
0=d(i_{X^c} \omega + c \theta) = d (i_{X^c} \omega) + c d \theta = \mathcal{L}_{X^c} \omega - c \omega.
\]

Thus, the result follows.
\end{proof}

\begin{proposition}
	\label{prop:local}
The flow $ \phi _t $ of a vector field $ X ^c $ satisfies $ \phi ^\ast _t \omega = e ^{ ct }  \omega $ (that is,  $ \phi _t $ is a conformally symplectic transformations for each $ t $)  if and only if $ X ^c $ is locally conformally Hamiltonian with parameter $ c $. 
\end{proposition} 
\begin{proof} 
If $X^c $ is a locally conformally Hamiltonian vector field then  
 \[\frac{ d}{dt} (\phi ^\ast _t \omega )  = \phi _t ^\ast \mathcal{L} _{ X ^c  } \omega = \phi _t ^\ast (c \omega)  \]
 which, as we have shown above, gives $ \phi _t ^\ast \omega = e ^{ ct } \omega $. Converserly suppose that $ \phi _t ^\ast \omega = e ^{ ct } \omega $. Then 
 \[\frac{ d}{dt} (\phi ^\ast _t \omega ) = \frac{ d } { dt } (e ^{ ct } \omega) =  e ^{ ct } (c\omega) = \phi _t ^\ast (c \omega)   \]
 then, since  $ \phi _t ^\ast (c \omega) =  \phi _t ^\ast \mathcal{L} _{ X ^c } \omega  $, we have  $ \mathcal{L} _{ X ^c } \omega = c\omega $. 
\end{proof} 

The following proposition offers a useful criterion for determining when a diffeomorphism is conformally symplectic. 
\begin{proposition}\label{prop:Jacobi-generalized}
	Let $ (M , \omega,\theta) $  an exact symplectic manifold and  $ f: M \to M $  a diffeomorphism. Then $ f $ is conformally symplectic with parameter $ c $ if and only if 
	\[ f ^\ast X _H =  \frac{ 1 } { c } X _{ H \circ f } \]
for all function $ H $ on $ M $. 
\end{proposition} 
\begin{proof}Suppose  $ f $ is conformally symplectic, that is, $ f ^\ast \omega = c\omega $, then 
	\begin{align*} d (H \circ f) (z) \cdot v & = \omega  (z) (X _{ H \circ f } (z) , v)  = \frac{ 1 } { c } (f ^\ast \omega) (z)  (X _{ H \circ f } (z) , v)\\
		& =  \frac{ 1 } { c } \omega (f(z))  (f _\ast X _{ H \circ f } (z) , f_\ast v)    \\
		& =  \omega (f(z))  \left( \frac{ 1 } { c } f _\ast X _{ H \circ f } (z) , f_\ast v\right)   
	\end{align*} 
On the other hand, by the chain rule 
\[ d (H \circ f) (z) \cdot v = dH (f (z)) \cdot [d f(z) \cdot  v ] = dH (f (z)) \cdot [f _\ast(z) v ] = \omega (f (z)) (   X _H (f (z)) ,   f _\ast v)     \]
Since $ \omega $ is non-degenerate it follows that $ \frac{ 1 } { c } f _\ast X _{ H \circ f } (z) =  X _H (f (z)) $, or 
\[
	f ^\ast X _H ( f (z)) = \frac{ 1 } { c } X _{ H \circ f } (z).
\]
The converse can be demonstrated by reversing the steps of the argument.
\end{proof} 
The following corollary follows directly from the above proposition, the fact that \( H \circ f = f^\ast H = bH \), and the principle of linearity.
\begin{corollary}\label{cor:f-ast-X_H}
	Let $ (M , \omega,\theta) $  an exact  symplectic manifold,   $ f: M \to M $    conformally symplectic with parameter $ c $, and $ H $ is a function on $ M $ that is conformally invariant with parameter $ b $ (i.e. $ f ^\ast H = b H  $) then 
	\[ f ^\ast X _H =  \frac{ b } { c } X _{ H  }. \]
\end{corollary} 

\section{Cotangent Bundles and Scaled Cotangent Lift}\label{sec:cotangent-bundle}


In many applications, especially in classical mechanics, the phase space is often modeled as the cotangent bundle \( M = T^\ast Q \) of a configuration space \( Q \). The cotangent bundle \( M \) is naturally equipped with a canonical 1-form, denoted \( \theta_0 \).

At any point \( z = (q, p) \in M \), where \( q \in Q \) and \( p \in T_q^*Q \), the natural projection map \( \pi: M \to Q \) takes a point  $ z $ in the cotangent bundle to its corresponding point $ q $  in the configuration space. The differential of this projection at \( z \), denoted \( d\pi_z: T_zM \to T_qQ \), maps a tangent vector \( v \in T_zM \) to a tangent vector \( d\pi_z(v) \in T_qQ \).

The  1-form \( \theta_0 \), evaluated at \( z \) and acting on \( v \in T_zM \), is defined as:
\[
(\theta_0)_z(v) = \langle p, d\pi_z(v) \rangle,
\]
where \( p \in T_q^*Q \) is the covector at the base point \( q \), and \( \langle p, d\pi_z(v) \rangle \) represents the natural pairing between the covector \( p \) and the tangent vector \( d\pi_z(v) \).

Alternatively, \( (\theta_0)_z \) can be expressed as:
\[
(\theta_0)_z = (d\pi_z)^*p,
\]
where \( (d\pi_z)^*: T_q^*Q \to T_z^*M \) is the  transpose of the map \( d\pi_z: T_zM \to T_qQ \). 


The one-form \( \theta_0 \) is called the \textbf{Liouville one-form} or \textbf{canonical one-form}. The exterior derivative \( \omega_0 = -d\theta_0 \) is a symplectic form on \( M \). The forms \( \theta_0 \) and \( \omega_0 \) are known as the \textbf{canonical forms} on \( M = T ^\ast Q \). In this paper, we will always consider \( T^\ast Q \) with the canonical forms, treating \( T^\ast Q \) as the exact symplectic manifold \( (T^\ast Q, \omega_0, \theta_0) \).

In the case of a cotangent bundle $T^\ast Q$ with coordinates $\mathbf{q}$ on $Q$ and $(\mathbf{q}, \mathbf{p})$ on $M = T^\ast Q$, the canonical forms are $\omega_0 = d\mathbf{q} \wedge d\mathbf{p}$ and $\theta_0 = \mathbf{p} \, d\mathbf{q}$. It is then possible to express the conformal Hamiltonian equations in a more concrete form, as shown in the next proposition.

\begin{proposition} Let $ M = T ^\ast Q $, and let  $ (\mathbf{q} , \mathbf{p})  $ be canonical coordinates for $ \omega_0 $, so the canonical forms are $ \omega_0 = d \mathbf{q}  \wedge d \mathbf{p} $  and $ \theta_0 = \mathbf{p} d \mathbf{q} $. Suppose  $ X _{ F } ^c $ is a conformally Hamiltonian vector field with parameter $ c $, and $F$ the associated conformal Hamiltonian. Then,  in these coordinates, we have 
	\begin{equation}\label{eqn:XF} X _{ F } ^c = \left( \frac{ \partial F } { \partial \mathbf{p} } , - \frac{ \partial F } { \partial \mathbf{q} } +c \mathbf{p}  \right) 
	\end{equation}
Moreover,  $ (\mathbf{q} (t) , \mathbf{p} (t)) $ is an integral curve of $ X _F ^c $ if and only if 
\[ \dot{\mathbf{q}}= \frac{ \partial F } { \partial \mathbf{p} } , \quad \dot{\mathbf{p}} = - \frac{ \partial F } { \partial \mathbf{q} } + c \mathbf{p} \]
\end{proposition}
\begin{proof}
	Let $ X _F ^c $ be defined by equation \eqref{eqn:XF}. We then have to verify that $ i _{ X _F ^c } \omega_0 + c \theta_0 = dF $.  Now $ i _{ X _F ^c } d \mathbf{q} = \frac{ \partial F } { \partial \mathbf{p} } $, $ i _{ X _F ^c } d \mathbf{p} = - \frac{ \partial F } { \partial \mathbf{q} } + c \mathbf{p}  $ by construction, so 
	\begin{align*} 
		i _{ X _F } ^c \omega_0 + c \theta_0  & = i _{ X _F ^c } (d \mathbf{q} \wedge d \mathbf{p} )= (i _{ X _F ^c}  d \mathbf{q}) \wedge d \mathbf{p}  - d \mathbf{q} \wedge ( i _{ X _F  ^c } d \mathbf{p})  + c \left( \mathbf{p}  d\mathbf{q}\right)   \\
						  & = \frac{ \partial F } { \partial \mathbf{p} } d \mathbf{p} + \left( \frac{ \partial F } { \partial \mathbf{q} } - c \mathbf{p}  + c \mathbf{p} \right) d \mathbf{q} = dF.
	\end{align*} 
\end{proof} 
Although the primary focus of this paper is not the study of dynamical systems with friction, it is worth noting that the formalism of conformal Hamiltonian systems can be effectively applied to describe such systems, as demonstrated in \cite{mclachlan2001conformal} (see also \cite{sloan2021scale}) and illustrated in the following example.


\begin{example}


Consider the following time-dependent Hamiltonian function:
\[ 
H(q, p, t) = e^{-bt} \frac{p^2}{2} + e^{bt} U(q), 
\]
which describes a dynamical system with friction. The Hamiltonian equations are:
\[ 
\dot{q} = \frac{\partial H}{\partial p} = e^{-bt} p, \quad \dot{p} = -\frac{\partial H}{\partial q} = -e^{bt} \frac{\partial U}{\partial q}, 
\]
so that Newton's equation of motion is:
\[ 
\ddot{q} = -b \dot{q} - \frac{\partial U}{\partial q}.
\]
Now we show that the same equation can be obtained using the conformal Hamiltonian formalism with a time independent Hamiltonian. 
Consider the function \( F = \frac{p^2}{2} + U(q) \). If this is a conformal Hamiltonian with parameter \( c = -b \), we find that the corresponding differential equations are:
\[ 
\dot{q} = \frac{\partial F}{\partial p} = p, \quad \dot{p} = -\frac{\partial F}{\partial q} = -bp - \frac{\partial U}{\partial q}, 
\]
which yields the same Newton's equation as the Hamiltonian \( H \). Moreover, since in this case  \( X_F^c = (p, -bp - \frac{\partial U}{\partial q}) \), \( \theta_0 = p \, dq \), and \( c = -b \), then by Proposition \ref{prop:dF/dt} we have: 
\[ 
\frac{dF}{dt} = -b i_{X_F^c} \theta = -bp^2. 
\]
\end{example}

We now introduce an important method for generating conformal symplectic transformations on cotangent bundles by generalizing the concept of cotangent lifts.

Let \( Q_1 \) and \( Q_2 \) be manifolds, and let \( f: Q_1 \to Q_2 \) be a diffeomorphism. Recall that the \textbf{cotangent lift} \( \hat{f}: T^*Q_1 \to T^*Q_2 \) is defined as follows:

For each point \( z_1 = (q_1, p_1) \in T^*Q_1 \), the cotangent lift \( \hat{f} \) maps \( z_1 \) to a point \( z_2 = (q_2, p_2) \in T^*Q_2 \), where:
\[
z_2=\hat{f}(z_1)  = \left( f(q_1), (df^{-1}_{f(q_1)})^*(p_1) \right),
\]
where 
\begin{itemize} 
	\item 	\( df^{-1}_{f(q_1)}: T_{f(q_1)}Q_2 \to T_{q_1}Q_1 \) is the differential of the inverse map \( f^{-1} \) evaluated at \( f(q_1) \), 
	\item \( (df^{-1}_{f(q_1)})^*: T_{q_1}^* Q_1 \to T_{f(q_1)}^* Q_2 \) is the transpose of $ df^{-1}_{f(q_1)}$, which maps the covector \( p_1 \in T_{q_1}^*Q_1 \) to a covector in \( T_{f(q_1)}^*Q_2 \).
\end{itemize} 

To extend this idea, we introduce a scaled version of the cotangent lift that includes a scaling factor:

\begin{definition}
Let \( Q_1 \) and \( Q_2 \) be manifolds, and let \( f: Q_1 \to Q_2 \) be a diffeomorphism. The \textbf{scaled cotangent lift} \( \hat{f}^\lambda : T^*Q_1 \to T^*Q_2 \), with a scaling parameter \( \lambda \in \mathbb{R} \), is defined by:
\[
\hat{f}^\lambda(z_1) = \left( f(q_1), \lambda \cdot (df^{-1}_{f(q_1)})^*(p_1) \right),
\]
for each \( z_1 = (q_1, p_1) \in T^*Q_1 \).
\end{definition}

In coordinates \((\mathbf{q}, \mathbf{p})\) on $ T ^\ast Q _1 $, the scaled cotangent lift is expressed as:
\begin{equation}\label{eqn:cotangent-lift-in-coordinates} 
\hat{f}^\lambda (\mathbf{q}, \mathbf{p}) = \left( f(\mathbf{q}), \lambda (Df(\mathbf{q}))^{-T} \cdot \mathbf{p} \right),
\end{equation}
where \( Df(\mathbf{q}) \) is the Jacobian matrix of \( f \) evaluated at \( \mathbf{q} \), and \( Df(\mathbf{q})^{-T} \) denotes its inverse transpose. This explicit coordinate expression simplifies the application of the scaled cotangent lift in concrete problems by providing a direct representation of the transformation.

This construction is important because it ensures that $ \hat f ^{ \lambda }  $ is conformally symplectic, as it is shown in the following two results.

\begin{proposition}\label{prop:canonical-one-form-preserved}
	Let \( f : Q_1 \to Q_2 \) be a diffeomorphism and let  \( \hat f ^{ \lambda }  : T^\ast Q_1 \to T^\ast Q_2 \) be its  scaled cotangent lift. Let   \( \theta_{Q_1} \) and \( \theta_{Q_2} \) be the canonical one-forms on \( T^\ast Q_1 \) and \( T^\ast Q_2 \), respectively. Then \( \hat f ^{ \lambda }  \) preserves the canonical one-forms up to a scale factor, that is,
\[ 
	 (\hat f ^{ \lambda }) ^\ast \theta_{Q_2} = \lambda \theta_{Q_1}, 
\]
where \( \lambda \) is a nonzero scalar.
\end{proposition}

\begin{proof}
We aim to show that \( ((\hat{f}^\lambda)^* \theta_{Q_2})_{z_1} = \lambda (\theta_{Q_1})_{z_1} \) for any point \( z_1 = (q_1, p_1) \in T^*Q_1 \).

We will use the following facts in the proof:

\begin{itemize}
    \item Let \( z_2 = \hat{f}^\lambda(z_1) = (q_2, p_2) \in T^*Q_2 \). By the definition of the scaled cotangent lift \( \hat{f}^\lambda \), we have:
    \[
    q_2 = f(q_1), \quad \text{and} \quad (df_{q_1})^\ast p_2 = \lambda p_1.
    \]
    
    \item The canonical 1-forms \( \theta_{Q_1} \) and \( \theta_{Q_2} \) are defined as:
    \[
    (\theta_{Q_1})_{z_1} = (d\pi_1)^\ast_{z_1} p_1, \quad (\theta_{Q_2})_{z_2} = (d\pi_2)^\ast_{z_2} p_2,
    \]
    where \( \pi_i: T^*Q_i \to Q_i \)  ( $ i = 1, 2 $)   are the natural projections \( \pi_i(q_i, p_i) = q_i \).
    
    \item For any \( z_1 = (q_1, p_1) \in T^*Q_1 \), we have:
    \[
    \pi_2 \circ \hat{f}^\lambda(z_1) = \pi_2(z_2) = q_2 = f(q_1) = f \circ \pi_1(z_1),
    \]
    so \( \pi_2 \circ \hat{f}^\lambda = f \circ \pi_1 \).
    
    \item By the definition of the pullback, we know:
    \[
    ((\hat{f}^\lambda)^* \theta_{Q_2})_{z_1} = (d\hat{f}^\lambda)^\ast_{z_1} (\theta_{Q_2})_{(z_2},
    \]
    where \( (d\hat{f}^\lambda)^\ast_{z_1}: T^*_{z_2}Q_2 \to T^*_{z_1}Q_1 \) is the transpose of the differential \( d\hat{f}^\lambda_{z_1} \).
\end{itemize}

Now, we can proceed with the calculation:

\begin{align*}
    ((\hat{f}^\lambda)^* \theta_{Q_2})_{z_1} & = (d\hat{f}^\lambda)^\ast_{z_1} (\theta_{Q_2})_{z_2} \\
    & = (d\hat{f}^\lambda)^\ast_{z_1} (d\pi_2)^\ast_{z_2} p_2 \quad \text{(since \( (\theta_{Q_2})_{z_2} = (d\pi_2)^\ast_{z_2} p_2 \))} \\
    & = (d(\pi_2 \circ \hat{f}^\lambda))^\ast_{z_1} p_2 \quad \text{(by the properties of the pullback)} \\
    & = (d(f \circ \pi_1))^\ast_{z_1} p_2 \quad \text{(since \( \pi_2 \circ \hat{f}^\lambda = f \circ \pi_1 \))} \\
    & = (d\pi_1)^\ast_{z_1} (df)^\ast_{q_1} p_2 \quad \text{(by the properties of the pullback)} \\
    & = (d\pi_1)^\ast_{z_1} (\lambda p_1) \quad \text{(since \( (df)^\ast_{q_1} p_2 = \lambda p_1 \))} \\
    & = \lambda (d\pi_1)^\ast_{z_1} p_1 \quad \text{(by linearity)} \\
    & = \lambda (\theta_{Q_1})_{z_1} \quad \text{(since \( (\theta_{Q_1})_{z_1} = (d\pi_1)^\ast_{z_1} p_1 \))}.
\end{align*}

Thus, we conclude that:
\[
(\hat{f}^\lambda)^* \theta_{Q_2} = \lambda \theta_{Q_1}.
\]
\end{proof}

\begin{corollary}\label{cor:symplectic-form-preserved}
	Let \( f : Q_1 \to Q_2 \) be a diffeomorphism and let  \( \hat f ^{ \lambda }  : T^\ast Q_1 \to T^\ast Q_2 \) be its  scaled cotangent lift. Let $ \omega _{ Q _1 } $ and $ \omega _{ Q _2 } $ be the  canonical symplectic forms on $ T ^\ast Q _1 $ and $ T ^\ast Q _2 $, respectively. Then $ \hat f  ^{ \lambda } $ is conformally symplectic, that is $ (\hat f  ^{ \lambda })^\ast \omega _{ Q _2 } = \lambda  \omega _{ Q _1 } $. 
\end{corollary}
\begin{proof}
 Let   \( \theta_{Q_1} \) and \( \theta_{Q_2} \) be the canonical one-forms on \( T^\ast Q_1 \) and \( T^\ast Q_2 \), respectively. By Proposition \ref{prop:canonical-one-form-preserved}
 we have that 	$(\hat f ^{ \lambda }) ^\ast \theta_{Q_2} = \lambda \theta_{Q_1}$. Taking the exterior derivative and changing the sign we get
 \[ -d ( (\hat f ^{ \lambda }) ^\ast \theta_{Q_2}) =- d ( \lambda  \theta_{Q_1}).\]
 Since the exterior derivative commutes with the pull-back, we get
 \[(\hat f ^{ \lambda }) ^\ast (- d ( \theta_{Q_2})) = \lambda (-d (   \theta_{Q_1})).\]
 The proof follows since $ \omega _{ Q _1 } = -d (\theta _{ Q _1 }) $ and $ \omega _{ Q _2 } = - d (\theta _{ Q _2 })  $.  
\end{proof}



\section{Conformally Symplectic Lie Group Actions}
Let \( G \) be a Lie group, and let \( \mathfrak{g} \) be its Lie algebra. Consider the action \( \Phi: G \times M \to M \) of \( G \) on \( M \). When \( g \) is fixed, we denote by \( \Phi_g: M \to M \) the map defined by \( \Phi_g(x) = \Phi(g, x) \). For each \( g \in G \), the map \( \Phi_g \) is a diffeomorphism of \( M \) whose inverse is given by \( \Phi_{g^{-1}} \). 

With these notations, we have that 
\[
\phi_t = \Phi_{\exp(t\xi)}
\]
is a flow on \( M \), where \( \exp: \mathfrak{g} \to G \) is the exponential map.\footnote{It is standard to use \( \exp(x) \) to denote the exponential map and \( e^x \) to represent the usual exponential function on real numbers. Since much of this paper concerns the multiplicative group \( \mathbb{R}^{+} \), we will alternate between \( \exp(x) \) and \( e^x \) to refer to the exponential map of the group  \( \mathbb{R}^{+} \), depending on which notation is more convenient in a given context.
}  To each element \( \xi \in \mathfrak{g} \), there is an associated vector field on \( M \), which we denote by \( \xi_M \), and call the {\bf infinitesimal generator} of the action corresponding to $ \xi $. This vector field is defined as follows:
\[
\xi_M(x) = \left. \frac{d}{d\tau} \right|_{\tau=0} \Phi_{\exp(\tau\xi)}(x).
\]

For a given point \( x \in M \), the group orbit through \( x \) is the set 
\[
\operatorname{Orb}(x) = \{\Phi_g(x) \mid g \in G\}.
\]
We denote by \( \operatorname{Ad}_g(\xi) \) the adjoint action of \( g \) on \( \xi \), which in the case of matrix groups is given by
\[
	\operatorname{Ad}_g(\xi) =  \left. \frac{ d } { dt } \right | _{ t = 0 } g\, \exp(t \xi) \,g ^{ - 1 }   =g \,\xi\, g^{-1}.
\]
A useful formula, needed later, is given in the following proposition. See Lemma 9.3.7 in \cite{marsden2013introduction} for a proof.

\begin{proposition}\label{prop:Ad}
    Let \( \Phi \) be a left action of a group \( G \) on \( M \). Then for every \( g \in G \) and \( \xi \in \mathfrak{g} \),
    \begin{equation}
    (\operatorname{Ad}_g \xi)_M = \Phi^*_{g^{-1}} \xi_M.
    \end{equation}
\end{proposition}

We now transition to the specific case where \( G = \mathbb{R}^+ \), the multiplicative group of positive real numbers.  This shift is important because much of this paper focuses on the action of  $ \mathbb{R}  ^{ + } $. 
In this setting, the Lie algebra \( \mathfrak{g} \) is identified with \( \mathbb{R} \), and the exponential map corresponds to the usual exponential function on real numbers.

We now introduce some important concepts for this work, namely, conformally symplectic actions and scaling symmetries. 
The following definition generalizes the notion of a symplectic action.

\begin{definition}
	Let \( (M, \omega, \theta) \) be an exact symplectic manifold. An action \( \Phi : \mathbb{R}^+ \times M \to M \) of the multiplicative group \( \mathbb{R}^+ \) on \( M \) is said to be \textbf{conformally symplectic} with parameter \( c \in \mathbb{R} \) if:
	\[
	\Phi_g^* \omega = g^c \omega \quad \text{for all} \quad g \in \mathbb{R}^+.
	\]
	In the special case when \( c = 0 \), the action is symplectic.
\end{definition}

Now, consider the Lie algebra element \( \xi \in \mathbb{R} \). The map \( \Phi_{e^{\tau \xi}} : M \to M \) defines a flow on \( M \). The \textbf{infinitesimal generator} of the action corresponding to \( \xi \) is the vector field:
\[
\xi_M = \left. \frac{d}{d\tau} \right|_{\tau = 0} \Phi_{e^{\tau \xi}}(z).
\]

Given the conformally symplectic property \( \Phi_g^* \omega = g^c \omega \), we can compute the Lie derivative of \( \omega \) with respect to \( \xi_M \) as follows:
\[
\mathcal{L}_{\xi_M} \omega = \left. \frac{d}{d\tau} \right|_{\tau = 0} \Phi_{e^{\tau \xi}}^* \omega = \left. \frac{d}{d\tau} \right|_{\tau = 0} e^{\tau c \xi} \omega = (c \xi) \omega.
\]
An action with this property is called \textbf{infinitesimally conformally symplectic}.

The following definition generalizes the concept of an invariant function.

\begin{definition}[Conformal Invariance]
Let \( (M, \omega, \theta) \) be an exact symplectic manifold, and let \( F: M \to \mathbb{R} \) be a function on \( M \). Suppose the action \( \Phi: \mathbb{R}^+ \times M \to M \) of the multiplicative group \( \mathbb{R}^+ \) on \( M \) is conformally symplectic with parameter \( c \). We say that \( F \) is \textbf{conformally invariant} if there exists \( b \in \mathbb{R} \) such that:
\[
\Phi_g^* F = g^b F \quad \text{for all} \quad g \in \mathbb{R}^+.
\]
If \( b = 0 \), the function is simply called invariant.
\end{definition}

If \( \xi_M \) is the infinitesimal generator associated with an element \( \xi \in \mathbb{R} \) from the Lie algebra, then the Lie derivative of \( F \) with respect to \( \xi_M \) is:
\[
\mathcal{L}_{\xi_M} F = \left. \frac{d}{d\tau} \right|_{\tau=0} \Phi_{e^{\tau \xi}}^* F = \left. \frac{d}{d\tau} \right|_{\tau=0} e^{b \xi \tau} F = (b \xi) F.
\]

This property is referred to as \textbf{infinitesimal conformal invariance}. It extends the classical notion of \textbf{infinitesimal invariance} (see \cite{holm2009geometric} for the classical definition). In the terminology of \cite{bravetti2023scaling}, a function \( F \) with this property is called a \textbf{scaling function}.

Scaling symmetries provide an important extension of Hamiltonian symmetries, where both the symplectic form and the Hamiltonian scale under group actions rather than remaining invariant (as in the case of Hamiltonian symmetries).  We formalize this concept in the following definition.


\begin{definition}[Scaling Symmetry]
Let \( (M, \omega, \theta, H) \) be a Hamiltonian system on an exact symplectic manifold, and let \( \Phi : \mathbb{R}^+ \times M \to M \) be a conformally symplectic action. We say that \( \Phi \) is a \textbf{scaling symmetry} if \( \Phi \) is conformally symplectic and the Hamiltonian \( H \) is conformally invariant. That is, there exist \( c, b \in \mathbb{R} \) such that:
\[
\Phi_g^* \omega = g^c \omega, \quad \Phi_g^* H = g^b H \quad \text{for all} \quad g \in \mathbb{R}^+.
\]
\end{definition}

If \( \xi_M \) is the infinitesimal generator of the scaling symmetry corresponding to the Lie algebra element \( \xi \in \mathfrak{g} = \mathbb{R} \), then:
\[
\mathcal{L}_{\xi_M} \omega = (\xi c) \omega, \quad \mathcal{L}_{\xi_M} H = (\xi b) H.
\]



The following proposition establishes a relationship between scaling symmetries and Hamiltonian vector fields. This result, which will be instrumental in our discussion of relative equilibria in Section \ref{sec:relative-equlibria}, is analogous to Corollary \ref{cor:f-ast-X_H} but applies specifically  to scaling symmetries rather than diffeomorphisms. Its proof follows a similar approach to that of Proposition \ref{prop:Jacobi-generalized} (see also \cite{mclachlan2018hamiltonian}).

\begin{proposition}\label{prop:scaling-X_H}
Let \( (M, \omega, \theta, H) \) be a Hamiltonian system on an exact symplectic manifold, and let \( \Phi : \mathbb{R}^+ \times M \to M \) be a scaling symmetry. Then:
\[
\Phi_g^* X_H = g^{b - c} X_H.
\]
\end{proposition}

\subsection{Scaled Cotangent Lifted  Action}
Given an action of  the multiplicative group \(  \mathbb{R}^{+} \) on a manifold \( Q \), we can define a corresponding action on the cotangent bundle \( T^*Q \). The standard method for this is using the cotangent lift \cite{holm2009geometric}. Here, however, we define a ``scaled" cotangent lift, which is more general.

\begin{definition}[Scaled Cotangent Lifted Action]
Let \( \Psi \) be an action of the group \( \mathbb{R}^{+} \) on a manifold \( Q \). For any \( g \in \mathbb{R}  ^{ + } \) and $ z = (p, q) \in T ^\ast Q $ , the scaled cotangent lift \( \hat{\Psi}_g ^c  \)  with parameter $ c \in \mathbb{R}  $ of the diffeomorphism \( \Psi_g \) to \( T^*Q \) is defined as:
\[  \hat\Psi_g ^c (z) = (\Psi _g (q) , g ^c \cdot (d \Psi _{ g ^{ - 1 } }) ^\ast (p) )     \]
where  \(\Psi_{g^{-1}} = (\Psi_g)^{-1} \).
Then the \textbf{scaled cotangent lifted action}  with parameter $ c $ is the action \( \hat{\Psi} ^c  : \mathbb{R}  ^{ + } \times T^*Q \to T^*Q \) given by 
\begin{equation}\label{eqn:cotangent-lifted-action-in-coordinates} 
\hat{\Psi} ^c (g, z) = \hat{\Psi}_g ^c (z). 
\end{equation} 
\end{definition}

If $ \Psi $ is the action of $ \mathbb{R}  ^{ + } $ on $ Q $ in  local coordinates $ (\mathbf{q} , \mathbf{p}) $ on $ T ^\ast Q $ , we can express the scaled cotangent lifted action with parameter $ c $  on $ T ^\ast Q $  as 
\begin{equation}\label{eqn:cotangent-lifted-action-coordinates}  \hat \Psi ^c  (g, (\mathbf{q} , \mathbf{p})) = \left( \Psi _g (\mathbf{q} ) , g ^c (D \Psi _g (\mathbf{q}) ) ^{ - T } \cdot \mathbf{p} \right).      
\end{equation} 

An important property of scaled lifted actions is given by the following proposition
\begin{proposition} Every scaled cotangent lifted action of the multiplicative group  $  \mathbb{R}  ^{ + } $ is conformally symplectic with respect to the canonical symplectic form $ \omega _0 $. 	 
\end{proposition} 
\begin{proof} 
The proof follows from Corollary \ref{cor:symplectic-form-preserved}.
\end{proof} 
The following proposition gives a useful  local coordinate formula  for the infinitesimal generators of the scaled  lifted action on \( T^\ast Q \).
\begin{proposition}[Infinitesimal Lifted Actions]
    Let \(  \mathbb{R}^+ \) act on \( Q \) via the action \( \Psi ^c : \mathbb{R}  ^{ + } \times Q \to Q \). Let \( \hat{\Psi} ^c :\mathbb{R}  ^{ + } \times T^\ast Q \to T^\ast Q \) denote the scaled cotangent lifted action on \( T^\ast Q \) with parameter \( c \). Let \( \xi \in  \mathbb{R} \) be an element of the Lie algebra of \( \mathbb{R}  ^{ + } \), and let \( \xi_Q \) be the infinitesimal generator of the action \( \Psi \) on \( Q \) associated with \( \xi \). Then, in local coordinates, the infinitesimal generator of the action \( \hat{\Psi} ^c \) on \( T^\ast Q \) is given by
    \begin{equation} \label{eqn:xi-ast-TQ}
        \hat{\xi}_{T^\ast Q} ^c  (\mathbf{q}, \mathbf{p}) = \left( \xi_Q (\mathbf{q}), \left( c \xi \operatorname{Id} - \left( D \xi_Q (\mathbf{q}) \right)^T \right) \cdot \mathbf{p} \right).
    \end{equation} 
\end{proposition}

	\begin{proof} 
Let \( \Psi \) be an action of \( \mathbb{R}^+ \) on \( Q \), and define \( g = \exp(t \xi) \). By definition, the infinitesimal generator \( \xi_Q \) is given by:
\[
\xi_Q(\mathbf{q}) = \left. \frac{d}{dt} \right|_{t=0} \Psi_{\exp(t \xi)}(\mathbf{q}).
\]
Now, using equation \eqref{eqn:cotangent-lifted-action-coordinates}, we compute the infinitesimal generator \( \hat{\xi}_{T^*Q}^c \) of the cotangent lift:

		\begin{align*} 
			\hat\xi _{ T ^\ast Q } ^c  (\mathbf{q} , \mathbf{p}) & = \left .\frac{ d }{dt}\right| _{ t = 0 } (\hat\Psi_{ \exp  t \xi } ^c  )(\mathbf{q} , \mathbf{p})  \\
								     & =  \left .\frac{ d }{dt}\right| _{ t = 0 }\left(  \Psi _{ \exp  t \xi }  (\mathbf{q})  , e ^{  c t \xi }   (D \Psi_{ \exp t \xi } (\mathbf{q})  )^{ - T }   \cdot \mathbf{p}  \right) \\
								     & = \left(  
								     \xi _{ Q } (\mathbf{q})  , \left .\frac{ d }{dt}\right| _{ t = 0 } e ^{ c t \xi}   (D \Psi_{ \exp  t \xi } (\mathbf{q})  )^{ - T }   \cdot \mathbf{p}  \right)\\
								     & = \left(  
								     \xi _{ Q } (\mathbf{q})  ,  
							     \left[ \left( \frac{ d }{dt} e ^{ c t\xi } \right)   (D \Psi_{ \exp  t \xi } (\mathbf{q})  )^{ - T }   \cdot \mathbf{p} +  e ^{ c t \xi}  \left( \frac{ d }{dt} (D \Psi_{ \exp  t \xi } (\mathbf{q})  )^{ - T }   \right)\cdot \mathbf{p} \right] _{ t = 0 }  \right)\\
								     & = \left(  \xi _{ Q } (\mathbf{q})  ,   \left( c\xi \operatorname{Id}  -  ( D \xi _{ Q } (\mathbf{q})  )^{ T }   \right)\cdot \mathbf{p}   \right)
		\end{align*}
		where the last equality follows from the fact that $\left .\Psi _{ \exp t \xi } (\mathbf{q}) \right| _{  t = 0 } = \mathbf{q} $  and the  following calculation. Since 
\[ \left\langle (D \Psi _{ \exp  t \xi} (\mathbf{ q })) ^{ - T } \cdot \mathbf{p} , D \Psi _{ \exp  t \xi } (\mathbf{q}) \cdot \mathbf{v} \right\rangle = \left\langle \mathbf{p} , \mathbf{v} \right\rangle      \]
for all $t$, differentiating both sides  with respect to $ t $ gives,
\[ \left . \frac{ d } { dt } \right | _{ t = 0 } \left\langle (D \Psi _{ \exp  t \xi} (\mathbf{ q })) ^{ - T } \cdot \mathbf{p} , D \Psi _{ \exp  t \xi } (\mathbf{q}) \cdot \mathbf{v} \right\rangle = 0\]
	which leads to: 
\begin{align*} 
	\left\langle \left. \frac{ d } { dt } \right| _{ t = 0 } (D \Psi _{ \exp  t \xi } (\mathbf{q})) ^{ - T } \cdot \mathbf{p} , \mathbf{v} \right\rangle & = - \left\langle \mathbf{p} , \left. \frac{ d } { dt } \right| _{ t = 0 } D \Psi _{ \exp  t \xi } (\mathbf{q}) \cdot \mathbf{v} \right\rangle \\
																			      & = - \left\langle \mathbf{p} , D \xi _{ Q } (\mathbf{q}) \cdot \mathbf{v} \right\rangle \\
																			      & = - \left\langle (D \xi _{ Q } (\mathbf{q})) ^T \cdot \mathbf{p} , \mathbf{v} \right\rangle.
\end{align*} 
Since \( \mathbf{v} \) is arbitrary, we conclude:
\[
\left. \frac{d}{dt} \right|_{t=0} (D \Psi_{\exp(t \xi)}(\mathbf{q}))^{-T} \cdot \mathbf{p} = -(D \xi_Q(\mathbf{q}))^T \cdot \mathbf{p}.
\]

	\end{proof} 

\section{Conformal Momentum Map}

\begin{definition}[{\bf Conformal momentum map}]
	Let $ (M, \omega,\theta) $ an exact symplectic manifold and $ \Phi : \mathbb{R}  ^{ + }   \times M \to M $ the action of the  multiplicative group  $\mathbb{R}  ^{ + } $  on $M$.  Suppose $ \Phi $ is a conformally symplectic action with parameter $ c $.	We say that a map $J: M \to \mathbb{R} $ 
	is a {\bf conformal momentum map}  with parameter $ c $ for  $ \Phi $ provided that for every $ \xi \in  \mathbb{R} $
	\[ i _{ \xi _M } \omega  + (c \xi)   \theta = dJ _{ \xi } , \]
	where $ J _{ \xi } : M \to \mathbb{R}  $ (the {\bf conformal momentum function} associated to $ \xi $) is defined by $ J _{ \xi } (z)   = \left\langle J,  \xi  \right\rangle =   \xi J (z) $, \footnote{
Here $ \left\langle \cdot , \cdot \right\rangle $ denotes the pairing between the Lie algebra and its dual. Since in this case the Lie algebra is  $ \mathfrak{g} = \mathbb{R} $, the pairing  reduces to the multiplication between real numbers.} and    $ \xi _M $ is the infinitesimal generator of the action corresponding to $ \xi$.
In other words, $ J $ is a conformal momentum map  for the action, provided that \begin{equation}\label{eqn:momentum-map-condition} 
	X _{J_\xi}^{\xi c} = \xi _M,
\end{equation} 
that is, $ \xi _M $ is a conformal Hamiltonian vector field with parameter $  \xi c $ and conformal Hamiltonian $ J _{ \xi } $. 
\end{definition}

\begin{remark}\label{rmk:momentum-map}
	 By the definition of conformal momentum map the function $ J _{ \xi } $ satisfies the equation $  i _{ \xi _M } \omega  + (c \xi)   \theta = dJ _{ \xi } $ with $X _{ J _{ \xi } } ^{ \xi c } =  \xi _M$. If we take $ \xi = 1 $, then it follows that $ J $ satisfies the equation  
	$ i _{ X ^{c} _{J }  } \omega  + c  \theta = dJ $, where $ X _J ^c $ is the infinitesimal generator corresponding to the Lie algebra element $1$.
\end{remark} 

The next result is the analogue of Theorem 4.2.10 in \cite{abraham2008foundations}, which helps in constructing conformal momentum maps. 
\begin{theorem}\label{thm:ixitheta} Let $ (M, \omega,\theta) $ an exact symplectic manifold. Let $ \Phi $ be a conformally symplectic action such that  $ \Phi_g^\ast \theta = g^c \theta $ for all $ g \in  \mathbb{R}  ^{ + } $. Then $ J : M\to  \mathbb{R}  ^{ + } $ defined by  
	\[ J _{ \xi } = \xi J  = i _{ \xi_M } \theta \]
is a conformal momentum map with parameter $ c $, that is, \( i _{ \xi _M } \omega  + (c \xi)   \theta = dJ _{ \xi } \). \end{theorem} 
\begin{proof} Let $ \phi _t = \Phi _{ e ^{ t \xi  } } $, be the flow associated with the group action $ \Phi $. Since $  \Phi_g ^\ast  \theta = g^c\theta $ it follows that 
\[ \mathcal{L} _{ \xi_M } \theta = \left. \frac{ d } { dt }  \right| _{ t = 0 }  \phi _t ^\ast \theta = \left. \frac{ d } { dt }  \right| _{ t = 0 }  \Phi _{ e 
		^{ t \xi }} ^\ast \theta =  \left. \frac{ d } { dt }  \right |_{ t = 0 }  e
	^{ (   t c \xi  ) }\theta = (c \xi)\,  \theta.   \]

Then by Cartan's formula 
\[  (c \xi) \,  \theta = \mathcal{L} _{ \xi _M } \theta = i _{ \xi _M } d\theta + d (i _{ \xi _M } \theta).\]
Hence, $ i _{ \xi _M } \omega + (c \xi) \,  \theta = d( i _{ \xi _M } \theta) $.  So, $ J _{ \xi }  = i _{ \xi _M } \theta $ satisfies the definition of momentum map. 

 \end{proof}

 \subsection{Conformal Momentum Maps for Actions on Cotangent Bundles}
Suppose \( Q \) is a manifold and \( M = T^\ast Q \) is the cotangent bundle with the canonical  forms. Let \( \Psi \) be an action of the multiplicative group \( \mathbb{R}^{+} \) on $ Q $, and let \( \hat{\Psi} ^c  \) be the scaled cotangent lifted action with parameter $ c $  of \( \Psi \) to \( T^\ast Q \). The following theorem provides a specialized formula for the conformal momentum map in this context, which will be useful in examples and applications.

\begin{theorem}[{\bf Noether's Formula for Scaled Cotangent Lifts}]\label{thm:Noether-formula}
Let \( \Psi \) be an action of the multiplicative group \( \mathbb{R}^+ \) on a manifold \( Q \), and let \( \xi \in \mathbb{R} \) be an element of the Lie algebra of \( \mathbb{R}^+ \). Suppose \( \hat{\Psi}^c \) is the scaled cotangent lift of \( \Psi \) with parameter \( c \). Then, the conformal momentum map \( J: T^*Q \to \mathbb{R} \) is given by:
\begin{equation}\label{eqn:J-cotangent-lifted}
J_\xi(z) = \left\langle J(z), \xi \right\rangle = \left\langle p, \xi_Q(q) \right\rangle,
\end{equation}
where \( z = (q, p) \in T^*Q \), with \( q = \pi(z) \) as the base point under the canonical projection \( \pi: T^*Q \to Q \), and \( p \in T^*_q Q \). Here, \( \xi_Q \) denotes the infinitesimal generator of the action \( \Psi \) on \( Q \), and the pairing on the right-hand side is between the covector \( p \in T_q^* Q \) and the vector \( \xi_Q(q) \in T_q Q \).

In local coordinates \( (\mathbf{q}, \mathbf{p}) \), the formula becomes:
\begin{equation}
J_\xi(\mathbf{q}, \mathbf{p}) = \xi J(\mathbf{q}, \mathbf{p}) = \left\langle \mathbf{p}, \xi_Q(\mathbf{q}) \right\rangle.
\end{equation}
\end{theorem}

\begin{proof}
We prove the theorem in local coordinates \( (\mathbf{q}, \mathbf{p}) \). In local coordinates, we have 
\[ 
X_{J_\xi}^{\xi c} = \left( \frac{\partial J_\xi}{\partial \mathbf{p}}, -\frac{\partial J_\xi}{\partial \mathbf{q}} + (c \xi) \mathbf{p} \right), 
\]
and so we need to show that  
\[ 
X_{J_\xi}^{\xi c} = \left( \frac{\partial J_\xi}{\partial \mathbf{p}}, -\frac{\partial J_\xi}{\partial \mathbf{q}} + (c \xi) \mathbf{p} \right) = \hat\xi_{T^\ast Q} ^c (\mathbf{q}, \mathbf{p}),
\]
where $ \hat \xi _{ T ^\ast Q } ^c $ is the infinitesimal generator of the scaled cotangent lifted action $\hat \Psi^c  $. 
Considering \( \mathbf{p} \) as a column vector, we have \( J_\xi(\mathbf{q}, \mathbf{p}) = \left\langle \mathbf{p}, \xi_Q(\mathbf{q}) \right\rangle = \mathbf{p}^T \xi_Q(\mathbf{q}) \). It follows that 
\[ 
\frac{\partial J_\xi}{\partial \mathbf{p}}(\mathbf{q}, \mathbf{p}) = \xi_Q(\mathbf{q}), 
\]
and 
\[ 
\frac{\partial J_\xi}{\partial \mathbf{q}}(\mathbf{q}, \mathbf{p}) \cdot \dot{\mathbf{q}} = \mathbf{p}^T D \xi_Q(\mathbf{q}) \dot{\mathbf{q}} = \left( \left( D \xi_Q(\mathbf{q}) \right)^T \mathbf{p} \right)^T \dot{\mathbf{q}}. 
\]
This implies
\[ 
-\frac{\partial J_\xi}{\partial \mathbf{q}}(\mathbf{q}, \mathbf{p}) + (c \xi) \mathbf{p} = \left( (c \xi) \operatorname{Id} - \left( D \xi_Q(\mathbf{q}) \right)^T \right) \mathbf{p}. 
\]
Hence, by Equation \eqref{eqn:xi-ast-TQ}, we have
\[
X_{J_\xi}^{\xi c}(\mathbf{q}, \mathbf{p}) = \left( \frac{\partial J_\xi}{\partial \mathbf{p}}, -\frac{\partial J_\xi}{\partial \mathbf{q}} + (c \xi) \mathbf{p} \right) = \left( \xi_Q(\mathbf{q}), \left( (c \xi) \operatorname{Id} - \left( D \xi_Q(\mathbf{q}) \right)^T \right) \mathbf{p} \right) = \hat\xi_{T^\ast Q} ^c (\mathbf{q}, \mathbf{p}).
\]
\end{proof}

The following example will be useful for our discussion of the Newtonian \( N \)-body problem.

\begin{example}\label{ex:momentum-map}
Consider the configuration space \( Q = \mathbb{R}^{3n} \) and its cotangent bundle \( T^*Q \) with local coordinates \( (\mathbf{q}, \mathbf{p}) \in T^*Q \). Let \( \Psi: \mathbb{R}^+ \times \mathbb{R}^{3n} \to \mathbb{R}^{3n} \) denote the group action of the multiplicative group \( \mathbb{R}^+ \) on \( Q \), defined by:
\[
\Psi(g, \mathbf{q}) = g \mathbf{q} \quad \text{for all } g \in \mathbb{R}^+.
\]
The corresponding scaled cotangent lifted action \( \hat{\Psi}^c \) is given by:
\[
\hat{\Psi}^c(g, (\mathbf{q}, \mathbf{p})) = \left( g \mathbf{q}, g^{c-1} \mathbf{p} \right).
\]
The choice \( c = 1/2 \) gives the Kepler scalings, which we will is useful when studying the Newtonian \( N \)-body problem. Clearly, the infinitesimal generator \( \xi_Q \) of the action on \( Q \) is:
\[
\xi_Q(\mathbf{q}) = \left. \frac{d}{dt} \right|_{t=0} \Psi_{\exp(t \xi)} = \xi \mathbf{q}.
\]
Thus, considering \( \mathbf{p} \) as a column vector, the conformal momentum function associated with \( \xi \) is:
\[
J_\xi(\mathbf{q}, \mathbf{p}) = \left\langle \mathbf{p}, \xi_Q(\mathbf{q}) \right\rangle = \xi \mathbf{p}^T \mathbf{q},
\]
and the momentum map is:
\[
J(\mathbf{q}, \mathbf{p}) = \mathbf{p}^T \mathbf{q}.
\]

As an exercise, we can compute \( \hat{\xi}_{T^*Q}^c(\mathbf{q}) \) using equation \eqref{eqn:xi-ast-TQ}. Explicitly, we have:
\[
\hat{\xi}_{T^*Q}^c(\mathbf{q}) = \left( \xi_Q(\mathbf{q}), \left( c \xi \operatorname{Id} - (D \xi_Q(\mathbf{q}))^T \right) \mathbf{p} \right) = \left( \xi \mathbf{q}, (c-1) \xi \mathbf{p} \right),
\]
which for \( c = 1/2 \) simplifies to:
\[
\hat{\xi}_{T^*Q}^c(\mathbf{q}) = \left( \xi \mathbf{q}, -\frac{1}{2} \xi \mathbf{p} \right).
\]

Next, we verify that \( J \) is indeed a conformal momentum map with parameter \( c = 1/2 \) by showing that \( \hat{\xi}_{T^*Q}^c \) is a conformally Hamiltonian vector field with conformal Hamiltonian \( J_\xi \). Explicitly, we compute:
\[
\frac{\partial J_\xi}{\partial \mathbf{p}} (\mathbf{q}, \mathbf{p}) = \xi \mathbf{q} = \xi_Q(\mathbf{q}),
\]
and for \( c = 1/2 \), we have:
\[
- \frac{\partial J_\xi}{\partial \mathbf{q}} + (c \xi) \mathbf{p} = -\xi \mathbf{p} + \frac{1}{2} \xi \mathbf{p} = -\frac{1}{2} \xi \mathbf{p}.
\]
Thus, we confirm:
\[
\left( \frac{\partial J_\xi}{\partial \mathbf{p}}, - \frac{\partial J_\xi}{\partial \mathbf{q}} + \frac{1}{2} \xi \mathbf{p} \right) = \hat{\xi}_{T^*Q}^c,
\]
which shows that \( J \) is a conformal momentum map, since \( J_\xi \) is a conformal Hamiltonian with parameter \( c = 1/2 \) for the conformally Hamiltonian vector field \( \hat{\xi}_{T^*Q}^c \).
\end{example}

 \subsection{Properties of Conformal Momentum Maps}
In this subsection, we explore key properties of conformal momentum maps in exact symplectic manifolds with conformally symplectic actions. We begin by generalizing Noether's theorem to the setting of conformally invariant Hamiltonians, leading to new conservation laws for conformal momentum maps. We then explore the connection between conformal momentum maps and scaling functions, establishing conditions under which the momentum map serves as a scaling function. Finally, we extend the notion of equivariance to conformal equivariance, which in this context reduces to conformal invariance. We demonstrate how the conformal momentum map exhibits this conformal invariance in scaled cotangent lifted actions.

With the definitions of conformal invariance and conformal momentum map now established, we can derive a generalization of Noether's theorem for this context. This generalized theorem extends the result presented in \cite{zhang2020generalized}, which is restricted to cotangent bundles.

\begin{theorem}[Hamiltonian Noether's Theorem]
    Let \( (M, \omega,\theta) \) be an exact symplectic manifold. Consider \( \Phi \) to be a conformally symplectic action with parameter $ c $  of \(\mathbb{R}^{+} \) on \( M \) with a conformal momentum map \( J \). If \( H \) is conformally invariant with parameter \( b \), then along the flow of \( X_H \), the momentum map \( J \) satisfies the equation:
    \[
    \frac{dJ}{dt} = -bH + c \, i_{X_H} \theta
    \]
    Consequently, the function 
    \[
    F = J + bHt - c \int_{0}^{t} i_{X_H} \theta \, dt
    \]
    is a constant of motion, that is,  \( \frac{dF}{dt} = 0 \).
\end{theorem}

\

\begin{proof}
	By Remark \ref{rmk:momentum-map}  $ J $ satisfies the equation  $ i _{ X ^{c} _{J }  } \omega  + c  \theta = dJ $, where $ X _J ^c $ is the infinitesimal generator corresponding to the Lie algebra element $1$. Hence,  
	\begin{align*} 
	\mathcal{L} _{ X _H } J & = dJ (X_H) = i _{ X _H } (dJ) \mbox{ and using } dJ =i _{ X_J^c } \omega + c \theta   \\
	& = i _{ X _H } i _{ X_J^c } (\omega) + c\,i _{ X _H } \theta \\ 
	& = - i _{ X_J^c } i _{ X _H } \omega + c \,i _{ X _H } \theta \\
	& = - i _{  X_J^c } dH + c \,i _{ X _H } \theta \\
	& = - \mathcal{L} _{ X_J^c } H  + c\, i _{ X _H } \theta\\
	& = - b H  + c\, i _{ X _H } \theta.
\end{align*} 
Hence,
\[ \frac{ dJ } { dt } = \{ J, H \} = \mathcal{L} _{ X _H } J = - b H  + c\, i _{ X _H } \theta,\]
where $ 
\{ \cdot, \cdot \} $ denotes the Poisson brackets. 

\end{proof} 

We now  establish an important property of conformally Hamiltonian vector fields. This result lays the foundation for understanding how conformal momentum maps relate to scaling functions:

\begin{lemma}\label{lem:iXtheta}
    Let \( (M, \omega, \theta) \) be an exact symplectic manifold. Suppose \( F: M \to \mathbb{R} \) is a smooth function, and let \( X_F^c \) be a conformally Hamiltonian vector field with parameter \( c \), with conformal Hamiltonian \( F \). Then the following statements hold:
    \begin{enumerate}
        \item If \( F = i_{X_F^c} \theta \), then \( \mathcal{L}_{X_F^c} \theta = c \theta \).
        \item If \( M \) is connected and \( \mathcal{L}_{X_F^c} \theta = c \theta \), then \( F = i_{X_F^c} \theta + \text{constant} \).
    \end{enumerate}
\end{lemma}
\begin{proof}
    Suppose \( F = i_{X_F^c} \theta \). By definition of conformally Hamiltonian vector field we have:
    \[
    dF = d(i_{X_F^c} \theta) = c \theta + i_{X_F^c} \omega.
    \]
Since $ \omega = - d \theta $, using Cartan's magic formula \( \mathcal{L}_{X_F^c} \theta = d(i_{X_F^c} \theta) + i_{X_F^c} d\theta \), we obtain:
    \[
    \mathcal{L}_{X_F^c} \theta = d(i_{X_F^c} \theta) -i_{X_F^c} \omega = c \theta.
    \]
    This proves the first direction.

    Conversely, suppose \( \mathcal{L}_{X_F^c} \theta = c \theta \). Using Cartan's formula again, we have:
    \[
    c \theta = d(i_{X_F^c} \theta) - i_{X_F^c} \omega.
    \]
    Rearranging, we get:
    \[
    d(i_{X_F^c} \theta) = c \theta + i_{X_F^c} \omega = dF,
    \]
    which implies \( F = i_{X_F^c} \theta + \text{constant} \), since \( M \) is connected, and any two functions with the same differential must differ by a constant on a connected domain.
\end{proof}
Next, we demonstrate that the conformal momentum map defined in Theorem  \ref{thm:ixitheta}, is, in fact, a scaling function under the conformally symplectic action. This result ties directly to the reduction framework studied in \cite{bravetti2023scaling}:

\begin{corollary}
	Let \( (M, \omega, \theta) \) be an exact symplectic manifold, and let \( \Phi \) be a conformally symplectic action with parameter \( c \). Consider the  momentum map \( J: M \to \mathbb{R} \), defined by \( J_\xi = \xi J = i_{\xi_M} \theta \) as described in Theorem \ref{thm:ixitheta}, where \( \xi_M \) is the infinitesimal generator of the action corresponding to the Lie algebra element \( \xi \in \mathbb{R} \). With this choice of momentum map, \( J \) is a scaling function with parameter \( c \).
\end{corollary}

\begin{proof}The conformal momentum map \( J \) satisfies the equation \( i_{X_J^c} \omega + c \theta = dJ \), where \( X_J^c \) is the infinitesimal generator corresponding to the Lie algebra element \( 1 \). Since \( J = i_{X_J^c} \theta \), by Lemma \ref{lem:iXtheta}, we know that \( \mathcal{L}_{X_J^c} \theta = c \theta \).

To confirm \( J \) as a scaling function, we compute:

 \begin{align*}
        \mathcal{L}_{X_{J}^c} J &= i_{X_{J}^c} d{J} \\
        &= i_{X_J^c} (d(i_{X_J^c} \theta)) \quad (\text{since }  J = i _{ X _J ^c } \theta ) \\
        &= i_{X_J^c} (\mathcal{L}_{X_J^c} \theta - i_{X_J^c} d\theta) \quad (\text{by Cartan's magic formula}) \\
	&= i_{X_J^c} (c\theta + i_{X_J^c} \omega)  (\text{since } \mathcal{L}_{X_J^c} \theta = c \theta \mbox{ and } \omega = - d \theta  )  \\
	&= c i_{X_J^c} \theta \quad   (\text{since }  i _{ X _J ^c } i_{X_J^c} \omega =0  ) \\
        &= c J.
    \end{align*}
Therefore, \( J \) satisfies \( \mathcal{L}_{X_J^c} J = c J \), and thus \( J \) is a scaling function with parameter \( c \).
\end{proof} 

We now extend our analysis by exploring the relationship between a conformally Hamiltonian vector field and its associated function. In particular, we show that when the vector field satisfies a specific scaling condition on the symplectic potential, the function can be adjusted to become a scaling function:

\begin{proposition}\label{prop:Ltheta=theta}
Let \( (M, \omega, \theta) \) be a connected exact symplectic manifold, and let \( F: M \to \mathbb{R} \) be a smooth function. Suppose \( X_F^c \) is a conformally Hamiltonian vector field with parameter \( c \), associated with \( F \). If \( \mathcal{L}_{X_F^c} \theta = c \theta \), then there exists a function \( \tilde{F} \), differing from \( F \) by at most an additive constant, which is a scaling function with parameter \( c \) for the vector field \(  X_F^c \).
\end{proposition}

\begin{proof}
    Suppose \( \mathcal{L}_{X_F^c} \theta = c \theta \). By Lemma \ref{lem:iXtheta}, we know that \( F = i_{X_F^c} \theta + k_1 \) for some constant \( k_1 \). Define \( \tilde{F} = F - k_1 \), so that \( \tilde{F} = i_{X_F^c} \theta \). Clearly, \( d\tilde{F} = dF = c \theta + i_{X_F^c} \omega \), and \( X_F^c = X_{\tilde{F}}^c \). 

    Then, we compute:
    \begin{align*}
        \mathcal{L}_{X_{\tilde{F}}^c} \tilde{F} &= i_{X_{\tilde{F}}^c} d\tilde{F} \\
        &= i_{X_F^c} (d(i_{X_F^c} \theta)) \quad (\text{since } X_F^c = X_{\tilde{F}}^c) \\
        &= i_{X_F^c} (\mathcal{L}_{X_F^c} \theta - i_{X_F^c} d\theta) \quad (\text{by Cartan's magic formula}) \\
        &= i_{X_F^c} (c\theta + i_{X_F^c} \omega) \\
        &= c i_{X_F^c} \theta \\
        &= c \tilde{F}.
    \end{align*}
    Therefore, \( \mathcal{L}_{X_{\tilde{F}}^c} \tilde{F} = c \tilde{F} \), as required.
\end{proof}


Suppose \( G \) is a group acting on \( Q \), with an induced cotangent lifted action on \( T^*Q \). It is well known that the cotangent lifted action admits an equivariant momentum map, which is expressed by the formula:
\[
\left\langle J(z), \xi \right\rangle = \left\langle p, \xi_Q(q) \right\rangle,
\]
where \( z = (q, p) \in T^*Q \), with \( q = \pi(z) \) as the base point under the canonical projection \( \pi: T^*Q \to Q \), and \( p \in T^*_q Q \).
In the Abelian case, the notion of equivariance simplifies to invariance, as the adjoint action is trivial. The equivariance of the momentum map is crucial, as it ensures that the momentum map is compatible with the group action, a necessary condition for performing symplectic reduction.

However, the conformal momentum map associated with the scaled cotangent lifted action, as described in equation \eqref{eqn:J-cotangent-lifted}, is not equivariant. Despite this, we have shown previously that the scaled cotangent lifted action is infinitesimally conformally invariant. As we will demonstrate below, the action is, in fact, fully conformally invariant. To establish this, we begin with a definition that extends the concept of equivariance.

\begin{definition}
Let $(M, \omega,\theta)$ be an exact symplectic manifold with a conformally symplectic action $\Phi$ (with parameter $c$) of the multiplicative group $\mathbb{R}^{+}$ on $M$, and let $J$ be a conformal momentum map for $\Phi$. Then $J$ is said to be {\bf conformally equivariant} if for all $g \in \mathbb{R}  ^{ + } $,
\[ 
J(\Psi_g(z)) = g^c \operatorname{Ad}^*_{g^{-1}}(J(z)), 
\]
where $\Psi$ is the action of $\mathbb{R}  ^{ + } $ on $Q$.
\end{definition}

Since  here the group is $ \mathbb{R}^{+}$, the adjoint action is trivial. Hence,
\[ 
J(\Psi_g(z)) = g^c \operatorname{Ad}^*_{g^{-1}}(J(z)) = g^c J(z), 
\]
which means that conformal equivariance reduces to conformal invariance. The proof that the conformal momentum map is conformally invariant can be achieved by a slight modification of the classical proof found in \cite{marsden2013introduction}.

\begin{theorem}
Let \( \Psi \) be an action of the multiplicative group \( \mathbb{R}^+ \) on \( Q \), and let \( \hat{\Psi}^c \) be the scaled cotangent lift of this action with parameter \( c \). Then, the conformal momentum map associated with \( \hat{\Psi}^c \) is conformally invariant.
\end{theorem}

\begin{proof}
Consider the diffeomorphism \( \hat{\Psi}_g^c: T^*Q \to T^*Q \) associated with the scaled cotangent lift action \( \hat{\Psi}^c \). Let \( z = (q, p) \in T^*Q \), where \( q = \pi(z) \) is the base point under the canonical projection \( \pi: T^*Q \to Q \), and \( p \in T^*_q Q \) is the corresponding covector. 
Then, since \[ \hat{\Psi}_g^c(z) = (\Psi_g(q), g^c \cdot (d \Psi_{g^{-1}})^*(p)), \] for any \( g \in \mathbb{R}^+ \), we have the following:

\begin{align*}
	  \left\langle J ( \hat{\Psi}_g^c (z)), \xi \right\rangle 
	  & = \left\langle g^c\,  (d \Psi_{g^{-1}})^*(p), \xi_Q(\Psi_g(q)) \right\rangle  \quad \mbox{By   Theorem \ref{thm:Noether-formula} and the definition of } \hat\Psi _g ^c \\
	  & = \left\langle p, g^c\, (d \Psi_{g^{-1}}) \xi_Q(\Psi_g(q)) \right\rangle \\
	  & = g^c \left\langle p, d \Psi_{g^{-1}} \circ \xi_Q \circ \Psi_g(q) \right\rangle \\
	  & = g^c \left\langle p, (\Psi_g^* \xi_Q)(q) \right\rangle \quad \text{(by definition of the pullback)} \\
	  & = g^c \left\langle p, (\operatorname{Ad}_{g^{-1}} \xi)_Q(q) \right\rangle \quad \text{(by Proposition \ref{prop:Ad})} \\
	  & = g^c \left\langle J(z), \operatorname{Ad}_{g^{-1}} \xi \right\rangle \quad \text{(by Theorem \ref{thm:Noether-formula})} \\
	  & = \left\langle g^c \operatorname{Ad}^*_{g^{-1}} J(z), \xi \right\rangle \\
	  & = \left\langle g^c J(z), \xi \right\rangle.
\end{align*}
Thus, we conclude that \( J(\hat{\Psi}_g(z)) = g^c J(z) \), confirming the conformal invariance of \( J \).
\end{proof}

\section{Relative Equilibria of Scaling Symmetries}\label{sec:relative-equlibria}

The definition of relative equilibria for scaling symmetries requires careful consideration, as traditional definitions in the context of symplectic actions (e.g., Definition 8.3 in \cite{libermann2012symplectic}) are not applicable in our scenario. As a consequence, we adopt the approach outlined in \cite{montaldi2000relative} to define relative equilibria.




\begin{definition}[{\bf Relative Equilibria of Scaling Symmetries}]\label{defn:relativeequilibrium}
    Let $(M, \omega,\theta, H)$ be a Hamiltonian system on an exact symplectic manifold. Suppose there is a conformal symplectic action $\Phi$ of the multiplicative group $\mathbb{R}^+$ on $M$. Assume that $H$ is conformally invariant under $\Phi$, implying that $\Phi$ represents a scaling symmetry of the Hamiltonian system.

    Let $z_e \in M$ and $\gamma: I \to M$ be the maximal integral curve (with $I$ being the maximal interval of definition) of the Hamiltonian system $(M, \omega,\theta, H)$ passing through $z_e$, with the initial condition $\gamma(0) = z_e$.

    Then $z_e \in M$ is called a {\bf relative equilibrium} (with respect to the scaling symmetry) of the Hamiltonian vector field $X_H$ if, for every $t \in I$, there exists $\eta(t) \in \mathbb{R}  ^{ + } $ such that $\eta(0) = \operatorname{Id}$ and $\gamma(t) = \Phi_{\eta(t)}(z_e)$.
\end{definition}

The following result provides several equivalent characterizations of relative equilibria for scaling symmetries.
\begin{theorem}[Relative Equilibrium Theorem for Scaling Symmetries]\label{thm:relative-equilibrium}
Under the hypotheses of Definition \ref{defn:relativeequilibrium}  the following statements are equivalent. 

\begin{enumerate}[(i)] 
	\item $ z _e $ is a relative equilibrium.
	\item The curve $ \gamma (t) $,  the maximal integral curve of the Hamiltonian system with $\gamma(0) = z_e$, is contained in the group orbit   $ \operatorname{Orb}(z _e) = \{ \Phi _g (z _e)\, |\, g \in \mathbb{R}  ^{ + } \}   $.
	\item There exists an element \(\xi\) in the Lie algebra \(\mathbb{R}  \) of the symmetry group \(\mathbb{R}  ^{ + } \) such that:\[
X_H(z_e) = \xi_M(z_e)
\]
where \(X_H\) is the Hamiltonian vector field generated by \(H\), and \(\xi_M\) is the infinitesimal generator of the group action on \(M\) corresponding to \(\xi\).
\item There exists a Lie algebra element  $ \xi \in \mathbb{R}  $ such that $ \gamma (t) = \Phi  _{ \eta (t) } (z _e) $ is a maximal integral curve of the Hamiltonian system $ (M , \omega , H) $ with 
	\[ \eta (t) =  \begin{cases}   [ (c - b) \xi t + 1 ]  ^{1/ (c - b) }& \mbox{if } b \neq c  \\e ^{ \xi t } & \mbox{if } b = c.\end{cases}
\]
\end{enumerate} 
\end{theorem} 
\begin{proof}
The logic of the proof will go as follows:
\[ (i) \iff (ii) , \quad \mbox{and}\quad  (i)\implies (iii)\implies (iv)\implies (i).\]

	$(i)\implies (ii) $.
	By definition, for each $ t \in I $ there exists $ \eta (t) \in 
	\mathbb{R}  ^{ + } $ such that  $ \eta (0)  = \operatorname{Id} $ and     $\gamma(t) =  \Phi_{\eta(t)}(z_e)$. Thus,  $\gamma(t)$ remains in the group orbit of $z_e$ for all $t \in \mathbb{R}$. 

$(ii)\implies (i) $. 
Let $\gamma: I \to \mathbb{R}  $ be the maximal integral curve of the Hamiltonian system with $\gamma(0) = z_e$. Then, by assumption, $\gamma(t)$ lies in $\operatorname{Orb}(z_e)$,   for all $t \in I$. 
   Hence, by definition of the group orbit, for each $ t \in \mathbb{R}  $ there is $ \eta(t) \in 
   \mathbb{R}  ^{ + } $ such that   \[
   \gamma(t) =\Phi_{\eta(t)}(z_e),
   \]
   that is, $ z _e $ is a relative equilibrium.

	$(i)\implies (iii) $. By the definition of relative equilibrium it follows that the trajectory $ \gamma (t)  $ through $ z _e $ is contained in a single group orbit $ \operatorname{Orb}(z _e) $, so its derivative $ X _H (z _e) = \dot\gamma (0) $ lies on the tangent space to the group orbit $ T _{ z _e }  (\operatorname{Orb } (z _e)) = \{ \xi _M (z _e) | \xi \in 
	\mathbb{R}  \} $, where $ \xi _M (z _e) $ is the  infinitesimal generator of the action at $ z_e $. Hence, there is $ \xi \in \mathbb{R}  $ such that $ X _H (z _e) = \xi _M (z _e) $.   

$(iii)\implies (iv) $. We assume that $ X _H (z _e) $ is equal to $ \xi _M (z _e) $, and we generalize the argument  used in Proposition 8.5 page 239 in \cite{libermann2012symplectic}.

First, consider the path \( y(t) = \Phi(\exp(t \xi), z_e) \), where \( \Phi: \mathbb{R}^+ \times M \to M \) is the group action. Since \( y(0) = z_e \), we want to verify that \( y(t) \) satisfies the ODE \( \frac{dy}{dt} = \xi_M(y(t)) \). The infinitesimal generator \( \xi_M(z_e) \) is given by:
\[
\xi_M(z_e) = \left. \frac{d}{d\tau} \right|_{\tau = 0} \Phi(\exp(\tau \xi), z_e).
\]
Then, using a shift in the time parameter \( s = t + \tau \), we have:

\begin{equation} 	\begin{aligned}\label{eqn:dy/dt} \frac{ d y (t) } { dt } & =   \left. \frac{ d } { ds }\right |  _{ s= t }  \Phi (\exp (s \xi) , z_e)   =\left. \frac{ d } { d\tau }\right |  _{ \tau= 0 }  \Phi (\exp ((t+\tau) \xi) , z_e) \\
				& = \left. \frac{ d } { d\tau }\right |  _{ \tau= 0 } \Phi(\exp(\tau \xi)  , \Phi (\exp (t \xi) , z_e)) \\
				& = \xi _M ( \Phi (\exp (t \xi) , z_e))\\
				& = \xi _M (y (t))   
\end{aligned}
			\end{equation} 
so $ y (t) $ is a path satisfying $ y (0) = z _e $ and    the ODE $ \frac{ d y } { dt } = \xi _M (y (t)) $.

Next, we want to show that \( \xi_M(y(t)) = d\Phi_{\exp(t \xi)}(\xi_M(z_e)) \). Recall that for a smooth map \( f: M \to N \), the differential \( df \) acts on a vector \( v \in T_x M \) as:
\[
df_x(v) = \left. \frac{d}{d\tau} \right|_{\tau = 0} f(\alpha(\tau)),
\]
where \( \alpha(\tau) \) is a path in \( M \) such that \( \alpha(0) = x \) and \( \left. \frac{d}{d\tau} \right|_{\tau = 0} \alpha(\tau) = v \). The differential   $ df _x $ is also often denoted by $ f _\ast (x)  $ and called pushforward.  Using this definition, we obtain:
\begin{equation}\label{eqn:xiM}
\begin{aligned}
\xi_M(y(t)) &= \left. \frac{d}{d\tau} \right|_{\tau = 0} \Phi(\exp(t \xi), \Phi(\exp(\tau \xi), z_e)) \\
            &= \left. \frac{d}{d\tau} \right|_{\tau = 0} \Phi_{\exp(t \xi)}(y(\tau)) \\
            &= d\Phi_{\exp(t \xi)}(\xi_M(z_e)).
\end{aligned}
\end{equation}					
We now seek a formula for $ X _H( y (t))    $. 

Since the action is conformally symplectic with parameter \( c \) and \( H \) is conformally invariant with parameter \( b \), it follows by Proposition \ref{prop:scaling-X_H}  that
\[
	((\Phi_g)^\ast X_H)(z_e) = g^{b - c} X_H(\Phi_g(z_e)).
\]

Applying \( (\Phi_g)_\ast \) to both sides of this equation, we obtain, for each \( g \in \mathbb{R}^+ \),
\[
(\Phi_g)_\ast \left(((\Phi_g)^\ast X_H)(z_e)\right) = (\Phi_g)_\ast \left( g^{b - c} X_H(\Phi_g(z_e)) \right).
\]

Since \( (\Phi_g)_\ast \circ (\Phi_g)^\ast \) acts as the identity, this simplifies to
\[
X_H(\Phi_g(z_e)) = g^{b - c} \cdot ((\Phi_g)_\ast X_H)(  \Phi_g(z_e)) = g^{b - c} \cdot d\Phi_g (X_H(z_e)),
\]
where we have expressed the pushforward as \( d\Phi_g \) for clarity.


Then, taking $ g = \exp (\xi t) $ we find   
\begin{equation} \begin{aligned} \label{eqn:XH} X _H (y (t))&  =   X _H (\Phi _{\exp \xi t } (z _e))\\
	& = \exp( (b - c)\xi t) \cdot d \Phi _{\exp \xi t } (X _H (z _e)) \\
	& = \exp( (b - c)\xi t)\cdot d \Phi _{\exp \xi t } (\xi_M(z _e))\\
	& = \exp( (b - c)\xi t)\cdot \xi _M (y (t)) 
\end{aligned}\end{equation}  
where we used that $ \xi _M (z _e) = X _H (z _e) $, and equation \eqref{eqn:xiM}.

If \( b = c \), then using equations \eqref{eqn:dy/dt} and \eqref{eqn:XH}, we find that
\[
\frac{d}{dt} y(t) = \xi_M(y(t)) = X_H(y(t)).
\]
Since \( \gamma(t) \) is an integral curve of the Hamiltonian system, we must have \( \frac{d}{dt} \gamma(t) = X_H(\gamma(t)) \). Therefore, the curves \( y(t) \) and \( \gamma(t) \) are maximal integral curves of the same differential equation, with the same initial condition. Consequently, \( y(t) \) and \( \gamma(t) \) coincide, which implies that
\[
\Phi_{\exp(t \xi)}(z_e) = y(t) = \gamma(t) = \Phi_{\eta(t)}(z_e), 
\]
and thus \( \eta(t) = e^{\xi t} \).


If $ b \neq c $, we need to show that $\gamma(t)$, the Hamiltonian integral curve, is a reparametrization of the group trajectory $y(t)$.
Using equation \eqref{eqn:dy/dt} and \eqref{eqn:XH} we obtain 
\[
	\frac{ d  } { dt }y(t) = \xi _M (y (t)) = \exp \left( (c - b) \xi t \right) \cdot X _H (y (t)),   
\]
while  $ \gamma (t) $ to be an integral curve of the Hamiltonian system we must have $	\frac{ d  } { d \tau}\gamma(\tau) =  X _H (\gamma (\tau))$. This shows that the group trajectory $y(t)$ is not directly an integral curve of the Hamiltonian vector field $X_H$, but rather follows a scaled version of it, with the scaling factor $\exp((c-b)\xi t)$.

To convert \( y(t) \) into an actual Hamiltonian trajectory, we introduce a new time variable \( \tau \), so that \( \gamma(\tau) = y(t(\tau)) \)  satisfies the equation: 
\[
\frac{d}{d\tau}  y(t(\tau)) = X_H( y(t(\tau))).
\]
We can achieve this by defining \( \tau \) such that:
\[
\frac{d\tau}{dt} = \exp((c - b) \xi t).
\]
Solving this differential equation with the initial condition \( \tau(0) = 0 \) gives:
\[
\tau = \frac{1}{(c - b) \xi} \left[ \exp\left((c - b) \xi t \right) - 1 \right].
\]
Now, inverting this expression to find \( t \) in terms of \( \tau \), we get:
\[
t = \frac{1}{(c - b) \xi} \ln\left[(c - b) \xi \tau + 1 \right].
\]
Thus, we can express the integral curve \( \gamma(\tau) \) as a reparametrization of \( y(t) \) using \( t(\tau) \):
\[
\gamma(\tau) = y(t(\tau)) = \Phi_{\exp(\xi t(\tau))}(z_e).
\]
Substituting \( t(\tau) \) into the expression for \( \gamma(\tau) \) given above, we obtain:
\[
\gamma(\tau) = \Phi_{\eta(\tau)}(z_e), \quad \text{where} \quad \eta(\tau) = \left[(c - b) \xi \tau + 1 \right]^{1/(c - b)}.
\]

This shows that \( \gamma(\tau) \) is indeed an integral curve of the Hamiltonian system, reparametrized from the scaled group trajectory \( y(t) \).

$(iv)\implies (i) $.  This is clear, since  we are given an explicit formula for $ \gamma (t) = \Phi _{ \eta (t) } (z _e) $.     

\end{proof} 

\begin{remark} A more explicit proof that  $ (i)\implies (iii) $ is provided below.
	On the one hand we have, 
	\begin{align*} \xi _M(z_e) = \left.\frac{ d } { dt }  \right|_{t=0}  [\Phi_{\exp(t \xi)}( z_e)] &  = \left.\frac{ d } { dt }  \right|_{t=0}  [\Phi(\exp(t \xi), z_e)] \\
				& = \left[ D _1 \Phi (\exp t \xi , z_e) \cdot \frac{ d } { dt} (\exp (t \xi)) \right] _{ t = 0 }\\
				& =    D _1 \Phi (	\operatorname{Id} , z_e)  \cdot \xi 	\end{align*} 
On the other hand we have, 
	\begin{align*} X_H(z_e) = \left.\frac{ d } { dt }  \right|_{t=0} \gamma(t) &  = \left.\frac{ d } { dt }  \right|_{t=0}  [\Phi(g(t), z_e)] \\
				& = \left[ D _1 \Phi (g(t), z_e) \cdot \frac{ d } { dt} g(t) \right] _{ t = 0 }\\
				& =    D _1 \Phi (	\operatorname{Id} , z_e)  \cdot  \frac{ d } { dt} g(0)	\end{align*} 
Hence $ X _H (z _e) =  \xi _M (z _e) $ for some Lie algebra element $ \xi \in \mathbb{R}  $ such that $ \xi = \frac{ d } { dt} g (0) $.   
\end{remark}

\begin{theorem} 
  Let $(M, \omega,\theta, H)$ be a Hamiltonian system on an exact symplectic manifold. Suppose there is a conformal symplectic action $\Phi$ of the multiplicative group $\mathbb{R}^+$ on $M$. Assume that $H$ is conformally invariant under $\Phi$, and that $ J $ is the conformal momentum map with parameter $c$ of this action. 
Then the  point $ z _e $ is a relative equilibrium if and only if there is a Lie algebra element $\xi \in \mathbb{R}  $ such that 
\[
	d H _{ \xi } +(c \xi) \theta = d (H - \xi J) +(c \xi) \theta   =0
\]
where $ H _{ \xi } $ is the {\bf augmented Hamiltonian}. 
\end{theorem} 
\begin{proof} 
First suppose $ z _e $ is a relative equilibrium. By Theorem \ref{thm:relative-equilibrium}
 $ X _H (z _e) = \xi _M (z _e) $ for some $ \xi \in \mathfrak{g} $. By the definition of the conformal momentum map we have that $ X _H (z _e) = X _{ J _{ \xi } } ^{ \xi c } (z _e) $. Since 
\[  i _{ X _H }  \omega = dH, \mbox{ and } i _{  X _{ J _{ \xi } }  ^{ \xi c } } \omega + (c \xi) \theta = d  J _{ \xi }\]
subtracting the two equations gives 
\begin{align*} 
	i _{ X _H }  \omega -  i _{  X _{ J _{ \xi } }  ^{ \xi c } } \omega - (c \xi) \theta = d (H-  \xi J) \\
	i _{ X _H -  X _{ J _{ \xi } }  ^{ \xi c }} \omega - (c \xi) \theta = d (H - \xi J).  
\end{align*} 
Since  $ X _H (z _e) = X _{ J _{ \xi } } ^{ \xi c } (z _e) $, the equations above give 
\[
	d H _{ \xi }+(c \xi) \theta  = d (H - \xi J) +(c \xi) \theta   =0
\]
which is the condition for relative equilibria.

If, on the other hand, \( dH_{\xi} + (c \xi) \theta = 0 \), then by reversing the computations above and using the nondegeneracy of the symplectic form \( \omega \), it follows that \( X_H - X_{J_{\xi}}^{\xi c} = 0 \). This implies, by the definition of the conformal momentum map, that \( X_H = \xi _M  \), that is, \(X _H \) points in the direction of the group orbit through \( z_e \).

\end{proof} 

We can now give a general definition of central configurations. 
\begin{definition}[Central Configurations]
    Let \( (T^\ast Q, \omega_0,\theta _0 ,  H) \) be a Hamiltonian system on \( T^\ast Q \). Let \( \Psi \) be an action of the multiplicative group \(  \mathbb{R}^+ \) on \( Q \), and let \( \Phi = \hat{\Psi} ^c  \) be the induced scaled cotangent lifted action on \( T^\ast Q \), which is conformally symplectic. Assume that \( H \) is conformally invariant under \( \Phi \) so  that \( \Phi \) is a scaling symmetry of the system. 

Now,  let \( z_e = (q_e, p_e) \) denote a covector based at \( q_e \in Q \). If \( z_e = (q_e, p_e) \) is a relative equilibrium for the scaling symmetry, then the point  \( q_e \in Q \) is called a {\bf central configuration}.
\end{definition}
The following example demonstrates how the concepts and definitions discussed above apply to the Newtonian $n$-body problem.
\begin{example}\label{ex:n-body}
    Consider the Newtonian \( n \)-body problem, which describes the motion of \( n \) point particles with masses \( m_i \in \mathbb{R}^+ \), positions \( q_i \in \mathbb{R}^3 \), and momenta \( p_i \in \mathbb{R}^3 \), where \( i = 1, 2, \ldots, n \). Let \( \mathbf{q} = (q_1, \ldots, q_n) \in Q = \mathbb{R}^{3n} \) and \( \mathbf{p} = (p_1, \ldots, p_n) \in \mathbb{R}^{3n} \) so that \( (\mathbf{q}, \mathbf{p}) \) are canonical coordinates on \( T^\ast Q \), and let \( M = \operatorname{diag}(m_1,m _1, m _1  , m  _2 , m _2, m _2   \ldots, m_n, m  _n , m _n) \).

    The Hamiltonian function for this problem is
    \[
        H = \frac{1}{2} \mathbf{p}^T M^{-1} \mathbf{p} + U(\mathbf{q}),
    \]
    where
    \[
        U(\mathbf{q}) = - \sum_{i<j} \frac{m_i m_j}{\|q_i - q_j\|}
    \]
    is a homogeneous function of degree \(-1\). Then Hamilton's equations  of motion are
\[ \mathbf{\ddot q} = -M ^{ - 1 } \frac{ \partial U(\mathbf{q} ) } { \partial \mathbf{q} }. \]
We assume that the center of mass of the particles is located at the origin, i.e., \( m_1 q_1 + m_2 q_2 + \cdots + m_n q_n = 0 \). This assumption simplifies the analysis by eliminating technical complications due to the non-trivial interaction between the scaling group \(\mathbb{R}^+\) and translations in \( \mathbb{R}^2 \). Together, translations and scalings form a semidirect product, so fixing the center of mass at the origin allows us to isolate the effects of scaling.

Consider the action of the group \( \mathbb{R}^+ \) on \( T^\ast Q \) given by
\[
\Phi(g, (\mathbf{q}, \mathbf{p})) = (g \mathbf{q}, g^{-1/2} \mathbf{p}).
\]

From Example \ref{ex:momentum-map}, we know that \( \Phi \) is a scaled cotangent lift, implying that \( \Phi \) is a conformally symplectic action. Furthermore, it is easy to verify that \( \Phi^\ast H = g^{-1} H \), showing that \( \Phi \) is a scaling symmetry, sometimes referred to as Kepler scaling \cite{bravetti2023scaling}. The momentum map for this action, as noted in Example \ref{ex:momentum-map}, is \( J = \mathbf{p}^T \mathbf{q} \).

    To obtain the equations for the central configurations, we first write the equation \( d(H - \xi J) + (c \xi) \theta_0 = 0 \), where $ \theta _0 $ is the canonical one-form,  explicitly in coordinates:
    \begin{align*}
	    d(H - \xi J) + (c \xi) \theta_0 &=\frac{\partial H}{\partial \mathbf{q}} d\mathbf{q}  + \frac{\partial H}{\partial \mathbf{p}} d\mathbf{p} - \xi \left( \frac{\partial J}{\partial \mathbf{q}} d\mathbf{q} + \frac{\partial J}{\partial \mathbf{p}} d\mathbf{p} \right) + c \xi \mathbf{p} d\mathbf{q} \\
        &= \frac{\partial U}{\partial \mathbf{q}} d\mathbf{q} + M^{-1} \mathbf{p} d\mathbf{p} - \xi \left( \mathbf{p} d\mathbf{q} + \mathbf{q} d\mathbf{p} \right) + c \xi \mathbf{p} d\mathbf{q} \\
        &= \left( \frac{\partial U}{\partial \mathbf{q}} - \xi \mathbf{p} + c \xi \mathbf{p} \right) d\mathbf{q} + \left( M^{-1} \mathbf{p} - \xi \mathbf{q} \right) d\mathbf{p} = 0.
    \end{align*}

    From this, we get the equations
    \[
        \frac{\partial U}{\partial \mathbf{q}} + (c - 1) \xi \mathbf{p} = 0, \quad \text{and} \quad \mathbf{p} = \xi M \mathbf{q}.
    \]
    Hence,
    \[
        \frac{\partial U}{\partial \mathbf{q}} + (c - 1) \xi^2 M \mathbf{q} = 0.
    \]

    From Example \ref{ex:momentum-map}, we know that we must choose \( c = \frac{1}{2} \), then
    \[
        \frac{\partial U}{\partial \mathbf{q}} = \frac{1}{2} \xi^2 M \mathbf{q},
    \]
which gives the classical equation for the central configurations. Using Hamilton's equations we see that 
$  \mathbf{\ddot q} = \frac{1}{2} \xi ^2 \mathbf{q} $.  
This indicates that the acceleration vector of each particle is directed towards the origin with the magnitude of the acceleration being proportional to the particle's distance from the origin. Moreover, from the fact that the potential is homogeneous of degree \(-1\), it is easy to show that
    \[
        \xi^2 = - \frac{2U(\mathbf{q})}{\mathbf{q}^T M \mathbf{q}}.
    \]
\end{example}

\section{Relative Equilibria of Scaling Symmetry for Cotangent Bundles}\label{sec:rel-eq-cot-bundle}
We now turn our attention to \textbf{simple mechanical systems}, which are defined on a manifold \( Q \) with a Lagrangian \( L : TQ \to \mathbb{R} \).
In this context  \( (q, v) \in TQ \) is an element of the tangent bundle, \( v \in T_qQ \) represents a tangent vector in the tangent space at \( q \), \( (q, p) \in T^*Q \) is an element of the cotangent bundle, and \( p \in T_q^*Q \) denotes an element of the cotangent space at  \( q \).

The Lagrangian is given by:
\[
L(q, v) = K(q, v) - U(q),
\]
where \( K(q, v) \) represents the kinetic energy, and \( U(q) \) denotes the potential energy.

Assume that \( Q \) is equipped with a Riemannian metric \( g \), which defines a smoothly varying collection of inner products \( \langle\langle \cdot, \cdot \rangle\rangle_q \) on each tangent space \( T_qQ \). This inner product induces a norm \( \| \cdot \|_q \), defined by:
\[
\|v\|^2 _q  = \langle\langle v, v \rangle\rangle_q, \quad v \in T_qQ.
\]
The kinetic energy is then expressed as:
\[
K(q, v) = \frac{1}{2} \|v\|_q^2,
\]
where \( v \in T_qQ \) is a tangent vector at \( q \in Q \).

Let \( L : TQ \to \mathbb{R} \) be a smooth Lagrangian function. At each point \( q \in Q \), we restrict \( L \) to the fiber over \( q \), namely the tangent space \( T_qQ \), giving a function:
\[
L_q : T_qQ \to \mathbb{R}.
\]
On each tangent space, we define the mapping:
\[
\mathbb{F}L_q : T_qQ \to T_q^*Q, \quad v \mapsto dL_q(v),
\]
where \( dL_q(v) \in T_q^*Q \) is the derivative (or differential) of \( L_q \) at \( v \). This map associates a velocity \( v \in T_qQ \) with a momentum \( p = dL_q(v) \in T_q^*Q \).

For simple mechanical systems, given $ q \in Q $,  this fiber-wise map satisfies:
\[
\left\langle \mathbb{F}L_q(v), w \right\rangle = \left. \frac{d}{ds} \right|_{s=0} L(q, v+ s w) = \left\langle\left\langle v, w \right\rangle\right\rangle_q,
\]
for all \(v,  w  \in T_qQ \), where the inner product on the right-hand side is  the one induced by the Riemannian metric.

These fiber-wise maps \( \mathbb{F}L_q \) collectively assemble into the global map known as the \textbf{Legendre transform} or \textbf{fiber derivative}:
\[
\mathbb{F}L : TQ \to T^*Q, \quad (q, v) \mapsto (q, p) = (q, \mathbb{F}L_q(v)).
\]

In local coordinates, the Legendre transformation takes the form:
\[
	\mathbb{F}L(\mathbf{q}, \mathbf{\dot{q}}) = \left(\mathbf{q},\mathbb{F}L _{ q } (\mathbf{q}, \mathbf{\dot{q}}) \right)=\left(\mathbf{q}, \frac{\partial L}{\partial \mathbf{\dot{q}}}\right),
\]
where \( \mathbf{q} \) are the position coordinates and \( \mathbf{\dot{q}} \) are the velocity coordinates.

Then the inner product on the fibers of $ T ^\ast Q $ induced by $ \left\langle \left\langle \cdot , \cdot \right\rangle   \right\rangle_q$ is defined by the expression
\[\left\langle \left\langle p , \tilde p\right\rangle \right\rangle _{ q}^\ast = \left\langle \left\langle \mathbb{F}L_q ^{ - 1 } (p) , \mathbb{F}L_q ^{ - 1 } ({\tilde p}) \right\rangle \right\rangle _q       \]
for all $ p  , \tilde p \in T_q ^\ast Q $. 
The Lagrangian system on $TQ $ corresponds, via the Legendre transform, to a Hamiltonian system on $ T ^\ast Q $ with Hamiltonian: 
\[ H (q,p) =  K (q,p) + U (q) , \]
where:
\[ K  (q,p) =   \frac{1}{2}  \left\langle \left\langle p , p \right\rangle \right\rangle  _q ^\ast  = \frac{1}{2} \left\langle \left\langle \mathbb{F}L_q ^{ - 1 } (p) , \mathbb{F}L_q ^{ - 1 } ({ p}) \right\rangle \right\rangle_q.  
\] 
Note that we use the same symbol \( K \) for both the function \( K(q, p) \) defined on the \( T^*Q \) and \( K(q, v) \) defined on \( TQ \), which is a  slight abuse of notation.

Let \( \Psi \) be the action of \( \mathbb{R}^{+} \) on \( Q \), the configuration manifold of a simple mechanical system, and let $ \hat{\Psi} $ be its cotangent lift.  Assume that \( \Psi \) satisfies the following conditions:

\begin{assumption}\label{a:1}
	The potential energy \( U(q) \) is conformally invariant with parameter \( b \) with respect to $ \hat{\Psi} $, meaning that \( \hat{\Psi}_g^* U(q) =  \Psi_g^* U(q) = U(\Psi_g(q)) = g^b U(q) \) for some \( b \in \mathbb{R} \).
\end{assumption}
\begin{assumption}\label{a:2} The kinetic energy \( K(q,p) \) is conformally invariant with parameter \( a \) with respect to $ \hat{\Psi} $, meaning that \( \hat\Psi_g^* K(q,p) =g^a K(q,p) \) for some \( a \in \mathbb{R} \).
\end{assumption}

We have the following proposition:
\begin{proposition}\label{prop:assumptions}
Under Assumption \ref{a:1} and \ref{a:2}, if the action \( \Psi \) induces a scaled cotangent lifted action \( \hat{\Psi}^c \) on \( T^*Q \), defined by 
\[\hat{\Psi}^c _g (q,p) = \left(\Psi _g (q) , g ^c \cdot (d \Psi _{ g ^{ - 1 } }) ^\ast  (p)\right),     \]
then \( \hat{\Psi}^c \) is a conformally symplectic action. In particular, if $c = (b-a)/2 $, 
then \( \hat{\Psi}_g^c \) satisfies
\[
(\hat{\Psi}_g^c)^* K(q,p) = K(\hat{\Psi}_g^c(q,p)) = g^b K(q,p).
\]
In this case, the Hamiltonian \( H \) is scaling invariant, meaning
\[
(\hat{\Psi}_g^c)^* H = g^b H.
\]
Thus, \( \hat{\Psi}_g ^c  \) represents a scaling symmetry.
\end{proposition}
\begin{proof} 
	The cotangent lifted action $ \hat{\Psi} _g $ is defined by 
\[\hat{\Psi} _g (q,p) = \left(\Psi _g (q) ,  (d \Psi _{ g ^{ - 1 } }) ^\ast  (p)\right).     \]
Hence, by Assumption \ref{a:2}, 
	\[ (\hat{\Psi} _{ g }) ^\ast K (q, p) = \frac{1}{2}  \left\langle \left\langle (d \Psi _{ g ^{- 1 } } ) ^\ast( p)  , (d \Psi _{ g ^{- 1 } } ) ^\ast (p) \right\rangle \right\rangle _q  ^\ast = \frac{1}{2}  g ^a \left\langle \left\langle  p,p \right\rangle \right\rangle _q ^\ast.  \]
On the other hand, if $ c = (b - a) /2 $, then  
\begin{align*} (\hat{\Psi} _{ g } ^c ) ^\ast K (q, p) & = \frac{1}{2}  \left\langle \left\langle  g ^c \cdot ( d \Psi _{ g ^{- 1 } } ) ^\ast( p)  , g ^c \cdot (d \Psi _{ g ^{- 1 } } ) ^\ast (p) \right\rangle \right\rangle _q  ^\ast \\
	& = \frac{1}{2}  g ^{a+2c} \left\langle \left\langle  p,p \right\rangle \right\rangle _q ^\ast  \\
	& =  \frac{1}{2}  g ^{b} \left\langle \left\langle  p,p \right\rangle \right\rangle _q ^\ast \\
	& = g ^b K (q , p). 
\end{align*} 
By Assumption 1, we also  have that \(\hat{\Psi}_g^\ast U = g^b U\). Hence,  \((\hat{\Psi}_g ^c )^\ast H = g^b H\). Moreover, it is clear that   $ \hat{\Psi}^c $ is conformally symplectic with parameter $ c $,  by the definition of scaled  cotangent lifted action. 
\end{proof} 
Since the action is conformally symplectic by Proposition \ref{prop:assumptions}, Theorem \ref{thm:Noether-formula} 
implies that the momentum map corresponding to the action \( \hat{\Psi}^c \) is the map \( J: T^*Q \to 
\mathbb{R}  \) given by
\[
\left\langle J(q,p), \eta \right\rangle = \left\langle p, \eta_Q(q) \right\rangle,
\]
where \( \eta_Q \) is the infinitesimal generator of the action on \( Q \) corresponding to \( \eta \in \mathbb{R}  \).

For each $ q \in Q $ the {\bf locked inertia tensor} is the map $ \mathbb{I} (q ) : 
\mathbb{R}  \to \mathbb{R}  $ given by
\[\left\langle \eta,\mathbb{I} (q ) \xi\right\rangle = \left\langle \left\langle \xi _Q (q ) , \eta _Q (q) \right\rangle \right\rangle_q     \]
for every $ \xi, \eta \in \mathbb{R}  $. 
Given $ \xi \in \mathbb{R}  $ the {\bf augmented potential} $U _{ \xi }:Q\to \mathbb{R}  $    is defined by 
\[U _{ \xi }  (q) = U (q ) - \frac{1}{2} \left\langle \xi, \mathbb{I} \xi  \right\rangle    \]
and the {\bf augmented kinetic energy} by 
\[K _{ \xi }(q,p) = \frac{1}{2} \left( \| p - \mathbb{F}L_q (\xi _Q (q )) \|_{q}^\ast \right)  ^{ 2 }.    \]

The next proposition shows that we can  write the augmented Hamiltonian in terms of the augmented potential and the augmented kinetic energy. 
\begin{proposition}
	$H _{ \xi }q,p) = K _{ \xi } (q,p) + U _{ \xi } (q,p) 
	$, where $  (q, p) \in T ^\ast Q $.  
\end{proposition}

\begin{proof} 
	By the definition of  $ \left\langle \left\langle \cdot , \cdot \right\rangle \right\rangle _{q}^\ast $ and the definition of locked inertia tensor 
	we have that 
	\[ \left\langle \left\langle \mathbb{F}L_q (\xi _Q (q)),   \mathbb{F}L_q (\xi _Q (q)) \right\rangle \right\rangle _{q}^\ast = \left\langle \left\langle \xi _Q (q) , \xi _Q( q)) \right\rangle \right\rangle_q    = \left\langle \xi ,\mathbb{I} (q) \xi \right\rangle.    \]
	By the definition of (conformal) momentum map and Legendre transform  we have that   
	\begin{align*} \left\langle J (q,p) , \xi \right\rangle & = \left\langle p , \xi _Q (q) \right\rangle = \left\langle \mathbb{F}L _q  (\mathbb{F}L _q  ^{ - 1 }  (p) ), \xi _Q (q) \right\rangle = \left\langle \left\langle
		\mathbb{F}L_q ^{ - 1 }  (p) , \xi _Q (q) \right\rangle \right\rangle_q \\
		& =   \left\langle \left\langle \mathbb{F}L_q ^{ - 1 } (p)   ,   \mathbb{F}L_q  ^{ - 1 } (\mathbb{F}L_q (\xi _Q (q))) \right\rangle \right\rangle_q   = \left\langle \left\langle p , \mathbb{F}L _q  (\xi _Q (q)) \right\rangle \right\rangle _{q}^\ast   
\end{align*} 
Using the previous equations and the expression for the conformal momentum map yields

\begin{align*} 
	K _{ \xi } (q,p) +&  U _{ \xi } (q,p) 
	= 	\frac{1}{2}( \|p  - \mathbb{F}L_q(\xi_{Q}(q ))\|^\ast _{q})^2 + U(q ) - \frac{1}{2} \langle \xi, \mathbb{I}(q ) \xi \rangle 
	\\											      & = \frac{1}{2} (\| p \|^\ast_{q})^2  - \langle\langle p , \mathbb{F}L_q(\xi_{Q}(q)) \rangle  \rangle_{q }^\ast  + \frac{1}{2} (\|\mathbb{F}L_q(\xi_{Q}(q))\|^\ast_{ q} )^2
	+ U(q) -  \frac{1}{2} \langle \xi, \mathbb{I}(q ) \xi \rangle 
	\\
				& =  \frac{1}{2} (\| p \|^\ast_{q})^2  - \langle J(q,p), \xi \rangle  +  \frac{1}{2} \langle \xi, \mathbb{I}(q ) \xi \rangle  + U(q) -   \frac{1}{2} \langle \xi, \mathbb{I}(q ) \xi \rangle 
				\\
			& =  \frac{1}{2} (\| p \|^\ast_{q})^2   - \langle J(q,p), \xi \rangle + U(q) 
			\\
									      & = H(q,p) - \langle J(q,p), \xi \rangle\\
									      & = H_{\xi}(q,p).						
\end{align*} 
\end{proof}

\begin{theorem}\label{thm:augmented-potential}
	Consider a simple mechanical system on \( Q \) with a Lagrangian function \( L: TQ \to \mathbb{R} \), and let \( \Psi \) be an action of the multiplicative group \( \mathbb{R}^+ \) on \( Q \) satisfying Assumptions \ref{a:1} and \ref{a:2}. This action induces a scaled cotangent lift \( \hat{\Psi}^c \) with \( c = (b - a)/2 \), which is a scaling symmetry. 

	Then, a point \( z =  (q, p) \in T^*Q \) is a relative equilibrium of the scaling symmetry if and only if the following conditions are satisfied:
	\begin{enumerate}
	    \item \( p  = \mathbb{F}L_q (\xi_Q(q)) \), and
	    \item \( q \) solves the equation
	    \[
	    dU_\xi(q) + (c \xi) \theta_L (\xi_Q(q)) = 0,
	    \]
	    where \( \theta_L = (\mathbb{F}L)^* \theta_0 \) is known as the \textbf{Lagrangian one-form}. For further details on the Lagrangian one-form, see \cite{marsden2013introduction}.
	\end{enumerate}
\end{theorem}

\begin{proof}
First notice that  \( \hat{\Psi}^c \) is a scaling symmetry by Proposition \ref{prop:assumptions}. Then, let \( (\mathbf{q}, \mathbf{p}) \) be local coordinates, and define
\[
\left( \frac{\partial L}{\partial \mathbf{\dot{q}}} \right)_{\mathbf{\xi}_Q} = \left( \frac{\partial L}{\partial \mathbf{\dot{q}}} \right)_{\mathbf{\dot{q}} = \mathbf{\xi}_Q(\mathbf{q})},
\]
where \(\mathbf{\xi}_Q\) is the infinitesimal generator of the group action. Then we have
\[
K_{\xi}(\mathbf{q}, \mathbf{p}) = \frac{1}{2} \left\| \mathbf{p} - \left( \frac{\partial L}{\partial \mathbf{\dot{q}}} \right)_{\mathbf{\xi}_Q} \right\|^2 _q ,
\]
where \(\|\cdot\|_q \) is the norm induced by the Riemannian metric, which depends on \(\mathbf{q}\). The Hamiltonian is then expressed as
\[
H_{\xi}(\mathbf{q}, \mathbf{p}) = K_{\xi}(\mathbf{q}, \mathbf{p}) + U_{\xi}(\mathbf{q}),
\]
where \(U_{\xi}(\mathbf{q})\) is the amended potential. Furthermore, the Liouville one-form is \(	\theta_0 = \mathbf{p} \, d\mathbf{q}\).


Since the following formula holds (see Theorem 3.6 in \cite{do1992riemannian}):
\[
\frac{ \partial } { \partial q _k } \left\langle \left\langle u,v \right\rangle\right\rangle _q    = \left\langle  \left\langle \nabla _{\partial _{ q _k }  } u, v \right\rangle \right\rangle _q  + \left\langle  \left\langle u, \nabla _{\partial _{ q _k } } v \right\rangle \right\rangle _q,
\]
where \( \nabla \) is the Levi-Civita covariant derivative, it follows that the derivative of \( K_\xi(\mathbf{q}, \mathbf{p}) \) with respect to \( q_k \) is:
\begin{equation} \label{eqn:dq}
	\frac{ \partial } { \partial q _k } K _{ \xi } (\mathbf{q} ,\mathbf{p} ) = \frac{1}{2}  \frac{ \partial } { \partial q _k }    \left\| \mathbf{p} - \left( \frac{\partial L}{\partial \mathbf{\dot{q}}} \right)_{\mathbf{\xi}_Q} \right\|^2 _q = 
	\left\langle \left\langle \nabla_{\partial {q _k }} \left( \mathbf{p} - \left( \frac{\partial L}{\partial \mathbf{\dot{q}}} \right)_{\mathbf{\xi}_Q} \right), \mathbf{p} - \left( \frac{\partial L}{\partial \mathbf{\dot{q}}} \right)_{\mathbf{\xi}_Q}  \right\rangle \right\rangle _q.
\end{equation}

Moreover, 
\begin{equation} \label{eqn:dp}
\begin{aligned} 
	\frac{ \partial } { \partial p _k } K _{ \xi } (\mathbf{q} ,\mathbf{p} )&  = \frac{1}{2}  \frac{ \partial } { \partial p _k }    \left\| \mathbf{p} - \left( \frac{\partial L}{\partial \mathbf{\dot{q}}} \right)_{\mathbf{\xi}_Q} \right\|^2 _q\\
										& = 
	\left\langle \left\langle 
	\frac{ \partial } { \partial p _k } \left( \mathbf{p} - \left( \frac{\partial L}{\partial \mathbf{\dot{q}}} \right)_{\mathbf{\xi}_Q} \right), \mathbf{p} - \left( \frac{\partial L}{\partial \mathbf{\dot{q}}} \right)_{\mathbf{\xi}_Q}  \right\rangle \right\rangle _q\\
										& = 	\left\langle \left\langle 
	\frac{ \partial } { \partial p _k }  \mathbf{p}, \mathbf{p} - \left( \frac{\partial L}{\partial \mathbf{\dot{q}}} \right)_{\mathbf{\xi}_Q}  \right\rangle \right\rangle _q
									\end{aligned} 
\end{equation} 

To compute the differential of the Hamiltonian, we have:

\begin{align*}
    dH_{\xi} + (c \xi)\theta &= \frac{\partial K_{\xi}}{\partial \mathbf{q}} \, d\mathbf{q} + \frac{\partial K_{\xi}}{\partial \mathbf{p}} \, d\mathbf{p} + \frac{\partial U_{\xi}}{\partial \mathbf{q}} \, d\mathbf{q} + (c \xi) \, \mathbf{p} \, d\mathbf{q} \\
			     & = \left[  \frac{\partial K_{\xi}}{\partial \mathbf{q}}  + \frac{\partial U_{\xi}}{\partial \mathbf{q}}  + (c \xi) \, \mathbf{p} \right] d \mathbf{q}   +\frac{\partial K_{\xi}}{\partial \mathbf{p}} \, d\mathbf{p} 
\end{align*} 
From which we obtain the equations
\[
	\frac{\partial K_{\xi}}{\partial \mathbf{p}} = 0, \quad \frac{\partial K_{\xi}}{\partial \mathbf{q}}  + \frac{\partial U_{\xi}}{\partial \mathbf{q}}  + (c \xi) \, \mathbf{p} =0
\]
By equation \eqref{eqn:dp},  $ 	\frac{\partial K_{\xi}}{\partial \mathbf{p}} = 0 $ implies  $\mathbf{p} - \left( \frac{\partial L}{\partial \mathbf{\dot{q}}} \right)_{\mathbf{\xi}_Q}=0$. Using equation \eqref{eqn:dq} this implies that $\frac{\partial K_{\xi}}{\partial \mathbf{q}}=0$.

From this, it follows that
\[
\mathbf{p} = \left( \frac{\partial L}{\partial \mathbf{\dot{q}}} \right)_{\mathbf{\xi}_Q}, \quad \text{and} \quad \frac{\partial U_{\xi}}{\partial \mathbf{q}}  + (c \xi) \, \left( \frac{\partial L}{\partial \mathbf{\dot{q}}} \right)_{\mathbf{\xi}_Q} = 0.
\]
Multiplying by $ d \mathbf{q} $ the second equation can also be written as $\frac{\partial U_{\xi}}{\partial \mathbf{q}} d \mathbf{q}   + (c \xi) \, \left( \frac{\partial L}{\partial \mathbf{\dot{q}}} \right)_{\mathbf{\xi}_Q} d \mathbf{q} = 0  $.
Since \( \frac{\partial L}{\partial \mathbf{\dot{q}}} \, d\mathbf{q} \) is the coordinate expression of the Lagrangian one-form \(\theta_L\) (see \cite{marsden2013introduction}), the proof follows.
\end{proof}

The proof of the theorem shows that the equations for a relative equilibrium under a scaling symmetry in local coordinates are:
\begin{equation}\label{eqn:rel-equilibria-local-coord}
\mathbf{p} = \left( \frac{\partial L}{\partial \mathbf{\dot{q}}} \right)_{\xi_Q}, \quad \text{and} \quad \frac{\partial U_\xi}{\partial \mathbf{q}} + (c \xi) \left( \frac{\partial L}{\partial \mathbf{\dot{q}}} \right)_{\xi_Q} = 0,
\end{equation}
while the equations for central configurations are:
\begin{equation} \label{eqn:central-config-local-coord}
\frac{\partial U_\xi}{\partial \mathbf{q}} + (c \xi) \left( \frac{\partial L}{\partial \mathbf{\dot{q}}} \right)_{\xi_Q} = 0.
\end{equation}
These equations are particularly convenient in concrete examples.

We can now apply this theory to the $ n $-body problem.

\begin{example} 
Let us revisit the Newtonian \( n \)-body problem, using the same notations as in Example \ref{ex:n-body}. We aim to derive the equation for central configurations using Equation \eqref{eqn:central-config-local-coord}.

The Lagrangian of the \( n \)-body problem is given by
\[ L(\mathbf{q}, \mathbf{\dot{q}}) = \frac{1}{2} \mathbf{\dot{q}}^T M \mathbf{\dot{q}} - U(\mathbf{q}), \]
and therefore,
\[ \mathbb{F}L (\mathbf{q}, \mathbf{\dot{q}}) = \left( \mathbf{q}, \frac{\partial L}{\partial \mathbf{\dot{q}}} \right) = \left( \mathbf{q}, M \mathbf{\dot{q}} \right). \]

The action of the multiplicative  group \( \mathbb{R}^{+} \) on \( Q \) is given by \( \Psi(g, \mathbf{q}) = (g, g \mathbf{q}) \). The cotangent lifted action on \( T^\ast Q \) is \( \hat\Psi(g, (\mathbf{q}, \mathbf{{p}})) = (g \mathbf{q}, g ^{ - 1 }  \mathbf{p}) \). Hence, $ (\hat{\Psi} _g ) ^\ast K (\mathbf{q} , \mathbf{p} ) = g ^{ - 2 } K ( \mathbf{q} , \mathbf{p}) $, and $ \hat{\Psi}_g  ^\ast U (\mathbf{q}) =\Psi_g  ^\ast U (\mathbf{q}) = g ^{ - 1 } U (\mathbf{q})$. 
By Proposition \ref{prop:assumptions} with $ b = - 1 $ and $ a = - 2 $ we take $ c = (b-a)/2 = 1/2 $, and hence 
the scaled cotangent lifted action is 
\[ \hat{\Psi} ^c (g, (\mathbf{q}, \mathbf{p})) =  (g \mathbf{q}, g^{1/2} g ^{ - 1 }  \mathbf{p})= (g \mathbf{q}, g^{-1/2} \mathbf{p}).\]
It is straightforward to verify that \( (\hat{\Psi}_g ^c )^\ast H = g^{-1} H \) and \( (\hat{\Psi}_g ^c )^\ast \omega _0  = g^{1/2} \omega _0  \), confirming that \( \hat{\Psi}^c  \) is a scaling symmetry.

In our case, we have
\[
\left( \frac{\partial L}{\partial \mathbf{\dot{q}}} \right)_{\xi_Q} = M \left( \mathbf{\dot{q}} \right)_{\mathbf{\dot{q}} = \xi_Q(\mathbf{q})},
\]
and since \( \xi_Q(\mathbf{q}) = \xi \mathbf{q} \), it follows that
\[
\left( \frac{\partial L}{\partial \mathbf{\dot{q}}} \right)_{\xi_Q} = M \xi_Q(\mathbf{q}) = M \xi \mathbf{q}.
\]
Therefore, by Equation \eqref{eqn:central-config-local-coord}, and noting that \( \left\langle \xi, \mathbb{I} \xi \right\rangle = \xi^2 \mathbf{q}^T M \mathbf{q} \), we obtain
\[
\frac{\partial U_\xi}{\partial \mathbf{q}} + (c \xi) \left( \frac{\partial L}{\partial \mathbf{\dot{q}}} \right)_{\xi_Q} = \frac{\partial}{\partial \mathbf{q}} \left( U(\mathbf{q}) - \frac{1}{2} \xi^2 \mathbf{q}^T M \mathbf{q} \right) + c \xi^2 M \mathbf{q} = 0.
\]

Taking \( c = \frac{1}{2} \), the equation simplifies to
\[
\frac{\partial U}{\partial \mathbf{q}} = (c - 1) \xi^2 M \mathbf{q} = \frac{1}{2} \xi^2 M \mathbf{q},
\]
which matches the result obtained in Example \ref{ex:n-body}.
\end{example} 

\section*{Acknowledgments}
We thank Shengda Hu, Corey Shanbrom,  David Sloane and Cristina Stoica for useful discussions and remarks concerning this paper.
\bibliographystyle{amsplain}

\bibliography{references}
\end{document}